\documentclass[11pt,letterpaper]{article}
\usepackage[margin=1in]{geometry}
\usepackage[T1]{fontenc}
\usepackage{lmodern,microtype}
\usepackage{amsmath,amssymb,amsthm,mathtools,booktabs,array,tabularx,enumitem}
\setlist{topsep=4pt,itemsep=2pt,parsep=0pt}
\usepackage[table]{xcolor}
\usepackage{tikz,placeins}
\usetikzlibrary{shapes.geometric}
\usepackage{hyperref}
\usepackage{algorithm}
\usepackage{algpseudocode}

\hypersetup{hypertexnames=false,colorlinks=true,linkcolor=blue!45!black,citecolor=blue!45!black,urlcolor=blue!45!black}

\newtheorem{theorem}{Theorem}[section]
\newtheorem{lemma}[theorem]{Lemma}
\newtheorem{proposition}[theorem]{Proposition}
\newtheorem{corollary}[theorem]{Corollary}
\theoremstyle{definition}
\newtheorem{definition}[theorem]{Definition}

\newcommand{\eps}{\varepsilon}
\DeclareMathOperator{\clip}{clip}
\tikzset{
 hzgood/.style={ellipse,draw=blue!55!black,fill=blue!4,
   line width=.7pt,minimum width=1.65cm,minimum height=.85cm,
   inner sep=4pt,align=center},
 hzbuyer/.style={ellipse,draw=orange!65!black,fill=orange!5,dashed,
   line width=.7pt,minimum width=1.65cm,minimum height=.85cm,
   inner sep=4pt,align=center},
 hzedge/.style={draw=black!65,line width=.65pt},
 hzcoef/.style={fill=white,inner sep=2pt}}

\title{Exploring the Tractability Frontier of Hylland--Zeckhauser Models}
\author{Ioannis Panageas\thanks{\texttt{ipanagea@ics.uci.edu}}, Avi Shein\thanks{\texttt{ashein1@uci.edu}},
Shayan Taherijam\thanks{\texttt{staherij@uci.edu}}, and Vijay V. Vazirani\thanks{\texttt{vazirani@ics.uci.edu}}\\[.8ex]
University of California, Irvine}
\date{}
\begin{document}
\pagenumbering{roman}
\maketitle
\begingroup
\let\small\normalsize
\begin{abstract}
The Hylland--Zeckhauser (HZ) mechanism provides a market-based approach to
one-sided matching, making efficient equilibrium computation central
to its implementation. Under bi-valued utilities, it admits a strongly
polynomial-time algorithm~\cite{vy}, whereas with four values per buyer, the problem
is PPAD-hard under inverse polynomial approximation~\cite{ccpy}. To understand the tractability boundary,
we study two extensions of the bi-valued case: a third linear utility
value and separable concave utilities with two common slopes.

For every fixed rational $a\in(0,1)$, we prove that computing a
normalized additive $\varepsilon_{*}(a)/N$-equilibrium in an $N$-buyer
HZ market is PPAD-complete for some constant $\varepsilon_{*}(a)>0$,
even when all utilities belong to the same prescribed set
$\{0,a,1\}$. This resolves the main question left open by~\cite{vy,ccpy}:
whether three utility values already suffice for PPAD-hardness. 
Assuming the PCP-for-PPAD conjecture, we strengthen the above-stated result 
to constant additive-error hardness.

For the concave extension, we compute an exact equilibrium in
strongly polynomial time with two marginal utility slopes common across goods for
each buyer and arbitrary rational buyer--good breakpoints.
With three common marginal slopes, suitable breakpoints
encode the three valued linear case, yielding PPAD-hardness.

Thus, our paper identifies two thresholds separating tractable and 
intractable cases of  HZ models with utility functions specified by utility values or marginal slopes.
In both cases, the tractability threshold lies between two and three.
\end{abstract}
\endgroup

\thispagestyle{empty}
\clearpage
\setcounter{tocdepth}{2}
\tableofcontents
\clearpage
\pagenumbering{arabic}

\section{Introduction}

Allocating indivisible goods without monetary transfers is a basic
problem in matching markets, such as assigning students to dormitory
rooms. The Hylland--Zeckhauser (HZ) mechanism~\cite{hz} addresses
this problem through a market for probability shares of goods.
There are $N$ buyers with cardinal utilities and $N$ indivisible
goods, each available in one copy. Each buyer receives one unit
of artificial money and chooses an affordable bundle of probability
shares of total quantity one to maximize her expected utility.
Equilibrium prices make these demands clear the market. The resulting
doubly stochastic allocation can be implemented as a lottery over
assignments in which each buyer receives exactly one good. No
monetary transfers are made. For linear utilities over probability
shares, the allocation is ex ante envy-free and, with
minimum-expenditure tie-breaking among optimal bundles, ex ante
Pareto-efficient~\cite{hz}.

At fixed prices, a buyer's demand in the linear model is a linear
program with a budget constraint and a unit-quantity constraint.
Finding prices that clear the market is the computational difficulty.
An optimal bundle may mix an expensive preferred good with a cheaper
less preferred good to satisfy both constraints. Vazirani
and Yannakakis~\cite{vy} gave a strongly polynomial exact
algorithm for bi-valued utilities, where each buyer has at most
two values, which may differ across buyers, and asked whether
$\{0,1/2,1\}$ already suffices for intractability.
Chen, Chen, Peng, and Yannakakis~\cite{ccpy}
proved PPAD-hardness of inverse polynomial approximation with
four values per buyer. Yan and Liu~\cite{yl} subsequently gave
a polynomial-time additive $1/e$ approximation for utilities in
$[0,1]$ and a $5\times5$ market over $\{0,1/2,1\}$ whose
exact equilibria all have an irrational price. Irrationality does
not, however, establish hardness of approximation; the three-value
case remained open.

We study two extensions of the tractable bi-valued domain.
The first keeps utilities linear and adds a third value. The
second allows separable concave utilities over fractional
allocations, with two marginal slopes common across goods for
each buyer and a separate breakpoint for every buyer--good pair.
Taking all breakpoints in $\{0,1\}$ recovers bi-valued linear
utilities, so both extensions contain the tractable case. We can
therefore distinguish the effect of an additional utility value
from that of diminishing returns.

\subsection{Our results}

\paragraph{Three common utility values.}
We use the additive approximation notion of~\cite{ccpy}.
A normalized additive $\eps$-equilibrium consists of
nonnegative prices with minimum zero and an allocation that
clears every good and gives every buyer total quantity one.
Each buyer $i$ spends at most $1+\eps$ and receives utility at
least $V_i(p)-\eps$, where $V_i(p)$ is her maximum utility over
unit bundles costing at most one at prices $p$. This notion
allows expenditure on zero-utility goods and does not impose
minimum-expenditure tie-breaking; see Definition~\ref{def:approx}.

\begin{theorem}\label{thm:main}
For every fixed rational $a\in(0,1)$, there is a positive rational
constant $\eps_*(a)$ such that computing a normalized additive
$\eps_*(a)/N$-equilibrium of an $N$-buyer HZ market is PPAD-complete
when all utilities belong to $\{0,a,1\}$.
\end{theorem}

Membership in PPAD follows from~\cite[Section~7]{vy}.
Their approximation notion imposes all the conditions in
Definition~\ref{def:approx}, together with an approximate
minimum-expenditure condition. Dropping that additional condition
gives a solution of our problem at the same error parameter.
Their upper bound applies to rational utilities in $[0,1]$ and
positive rational accuracy, including $\eps_*(a)/N$.

The theorem answers the question of~\cite{vy} for
$\{0,1/2,1\}$ and, more generally, for every prescribed middle
value $a$. The same three values are used by all buyers; this is stronger than allowing
each buyer her own three values. In contrast, the hardness of \cite{ccpy} does not assume 
four common values for all buyers. 
Furthermore, every row and column of our constructed utility matrix
has only a bounded number of positive entries, with the bound
depending on $a$. Since $1/\eps=N/\eps_*(a)$ is polynomial
in the market size, Theorem~\ref{thm:main} rules out a fully
polynomial-time approximation scheme unless
$\mathrm{PPAD}\subseteq\mathrm{P}$.

\paragraph{Constant approximation and PCP for PPAD.}
The PCP-for-PPAD conjecture of Babichenko, Papadimitriou, and
Rubinstein~\cite{bpr} concerns generalized circuits in which a
small constant fraction of gates may fail. We use the formulation
of~\cite{pcpmarkets}: for some positive constant gate accuracy
and failure fraction, finding such an assignment is PPAD-hard.
The original conjecture asks for a quasilinear reduction. OpenAI
has announced a proof of that stronger formulation~\cite{openaipcp};
we state our constant-approximation result conditionally on PCP
for PPAD.

\begin{theorem}[Constant-error hardness under PCP for PPAD]\label{thm:pcp-intro}
Assuming the PCP-for-PPAD conjecture, for every fixed rational
$a\in(0,1)$ there is a positive rational constant
$\eps_0=\eps_0(a)$ such that computing a normalized additive
$\eps_0$-equilibrium of an HZ market with utilities in
$\{0,a,1\}$ is PPAD-hard.
\end{theorem}

\paragraph{Two common marginal slopes.}
For the concave extension, buyer $i$ has separable utility
$U_i(x_i)=\sum_j f_{ij}(x_{ij})$, where
\begin{equation}\label{eq:two-slopes-intro}
 f_{ij}(t)=c_i^{(2)}t+(c_i^{(1)}-c_i^{(2)})\min\{t,\tau_{ij}\},
 \qquad 0\le t\le1.
\end{equation}
The slopes $c_i^{(1)}\ge c_i^{(2)}\ge0$ are common across
goods for buyer $i$, while $\tau_{ij}\in[0,1]$ is the
breakpoint for good $j$. Budgets, supplies, and required bundle
quantities are still one.

\begin{theorem}[Two marginal slopes]\label{thm:two-slopes}
For rational input in~\eqref{eq:two-slopes-intro}, an exact
normalized HZ equilibrium can be computed in strongly polynomial time. Each buyer receives a minimum-expenditure
bundle among her utility-maximizing affordable bundles.
\end{theorem}

With three common slopes, zero-length segments can represent each
of the linear utilities $0$, $at$, and $t$. Thus the linear
hardness theorem immediately gives the corresponding concave
hardness result.

\begin{corollary}[Three common marginal slopes]\label{cor:three-slopes}
For every fixed rational $a\in(0,1)$, computing a normalized
additive $\eps_*(a)/N$-equilibrium is PPAD-hard for separable
concave utilities with the same ordered marginal slopes $1,a,0$
for every buyer--good pair, allowing segments of zero length.
\end{corollary}

Table~\ref{tab:frontier} summarizes these results.

\begin{table}[!htbp]
\centering
\caption{Complexity under restrictions on linear utility values and
common marginal slopes, for fixed rational $a\in(0,1)$.
In the concave model, slopes are common across goods for each buyer;
the three-slope hardness uses the same slopes for all buyers.
Zero-length segments are permitted.}
\label{tab:frontier}
\renewcommand{\arraystretch}{1.2}
\begin{tabularx}{\textwidth}{|>{\centering\arraybackslash}p{.31\textwidth}|>{\centering\arraybackslash}X|>{\centering\arraybackslash}p{.20\textwidth}|}
\hline
\rowcolor{blue!10}
\textbf{Utility restriction} & \textbf{Complexity and accuracy} & \textbf{Reference}\\
\hline
Linear, two values per buyer & Strongly polynomial, exact & \cite{vy}\\
\hline
Linear, common values $\{0,a,1\}$ & PPAD-complete, error $\eps_*(a)/N$ & Theorem~\ref{thm:main}\\
\hline
Linear, common values $\{0,a,1\}$ & PPAD-hard, constant error under PCP for PPAD & Theorem~\ref{thm:pcp-intro}\\
\hline
Linear, four values per buyer & PPAD-hard, inverse polynomial error & \cite{ccpy}\\
\hline
Concave, two common slopes & Strongly polynomial, exact & Theorem~\ref{thm:two-slopes}\\
\hline
Concave, three common slopes & PPAD-hard, error $\eps_*(a)/N$ & Corollary~\ref{cor:three-slopes}\\
\hline
\end{tabularx}
\end{table}

\FloatBarrier

\subsection{Technical overview}

We reduce from the bounded-degree threshold game~\cite{thresholdgames,purecircuit},
with a fixed small tolerance $\kappa>0$. Each vertex has a value
$z_v\in[0,1]$, which must be at most
$\kappa$ when the incoming sum exceeds $1/2+\kappa$, and at
least $1-\kappa$ when that sum is below $1/2-\kappa$.
Inside this tolerance interval, any value is allowed.
Both indegree and outdegree are at most two, and every vertex has
at least one outgoing edge. For fixed rational $a\in(0,1)$, we
choose a small constant $\tau$ and a large admissible integer $K$
before receiving the graph. All vertex and wire gadgets then have
constant size, with vertex sizes depending only on the bounded
degree pair.

Chen, Chen, Peng, and Yannakakis~\cite[Section~3, full version]{ccpy}
also reduce from threshold games, using at most four utility values
per buyer, which may differ across buyers.
Their edge gadgets use this freedom to cancel unwanted price
dependencies. With the same prescribed set $\{0,a,1\}$ for every
buyer, we cannot tune these comparisons separately. Our construction
uses two successive comparisons to obtain a price ratio $a^2$.
Its separation from the utility ratio $a$ lets wire buyers transmit
an approximately affine signal. Their spending also changes at the
source; degree-dependent supplies cancel that second change.
We first explain the reduction at an exact equilibrium.

\paragraph{Encoding values in prices.}
Each vertex $v$ has three target groups of goods $L_v,M_v,H_v$,
with prices $q_v,m_v,p_v$. Goods within a group have identical
utility vectors, allowing price equalization
(Lemma~\ref{lem:copies}). Lower-link buyers compare utilities $a$
and one on $L_v,M_v$; upper buyers make the same comparison on
$M_v,H_v$. Two reference groups of goods, supported by their own
anchor buyers, bound the target prices. Floor buyers can buy the
low and middle targets or their reference; upper buyers can buy
the middle and high targets or a second reference.

Suitable group sizes force
\[
 q_v=am_v=a^2p_v,\qquad q_v<1<m_v<p_v.
\]
Thus the two comparisons produce the ratio $a^2$, while a unit
of the low target is affordable and a unit of either other target
is not. The reference groups confine every high price to the same
interval $[L,U+\eta]$, of width $O_a(1/K)$, just below $a^{-2}$.
Writing $w=U-L$ and $\eta=\tau w$, we decode
$z_v=\clip(t_v)$, where $t_v=(p_v-L)/w$ and $\clip$
truncates to $[0,1]$. Then $0\le t_v\le1+\tau$ and
$|t_v-z_v|\le\tau$.
Figure~\ref{fig:vertex} shows the gadget;
Sections~\ref{sec:construction} and~\ref{sec:exact} give the group
sizes and prove the price relations.

\paragraph{Directed signals and cancellation.}
For each edge $u\to v$, introduce $K$ wire buyers valuing
$L_u$ at $a$ and $H_v$ at one. Since both high prices lie in
the same narrow interval, $q_u/p_v=a^2+O_a(1/K)<a$.
The low good offers better utility per dollar, whereas the high
good offers better utility per unit. Each wire buyer therefore
mixes the two, exhausting both budget and quantity.
With our calibration of the price interval, their total spending
on $H_v$ is
\[
 I_{uv}=\frac{Kp_v(1-a^2p_u)}{p_v-a^2p_u}
       =3-2t_u+O_a(1/K).
\]
The error is proportional to $p_u-p_v$, with a factor bounded
in terms of $a$ alone, so the bound holds for every edge;
see Lemma~\ref{lem:demand} and~\eqref{eq:wire-expansion}.
Spending on $H_v$ therefore records the value encoded at $u$.

The same buyers spend their remaining budget
$K-I_{uv}=K-3+2t_u+O_a(1/K)$ on $L_u$.
Thus outgoing wires contribute a term $2o_vt_v$ to revenue at
$v$, where $o_v$ is its outdegree. We choose the target supplies
so that the revenue required to clear them contains exactly the
same term. Subtracting wire spending cancels this extra dependence
on $t_v$. The numbers of link buyers adjust the constant terms
according to the indegree and outdegree, placing the threshold at
$1/2$. The remaining market-clearing identity is
\[
 \Phi_v=2\sum_{u\to v}t_u-1+O_a(1/K),
\]
where $\Phi_v$ is floor spending on the targets minus upper
spending on the upper reference goods.

If $\Phi_v>0$, floor buyers must spend on targets. This is
optimal only near the low endpoint, forcing $z_v\le\tau$.
If $\Phi_v<0$, upper buyers must spend on their reference,
forcing the high price to be at least $U$ and hence $z_v=1$.
Small $\tau$ and large $K$ therefore turn a large incoming sum
into a small output, and a small incoming sum into a large output,
whenever the sum lies outside the tolerance interval.
Proposition~\ref{prop:exact} formalizes this decoding argument.

\paragraph{Robustness and inverse-linear accuracy.}
Approximate buyers may spend on goods they value at zero, supplying
revenue absent from the exact argument. Even a small quantity can
matter if its price is high, so we must control expenditure.
We do this before assuming approximate price relations.
Our demand estimates bound purchases from inferior options while
accounting for both the budget and the unit-bundle constraint.
The construction also keeps each buyer's optimal utility a fixed
amount below one. Together with revenue bounds, these estimates
give total unapproved expenditure $\mathcal W=O_a(N\eps)$ for
sufficiently small $K\eps$.

When $\mathcal W$ and $K^2\eps$ are small, approximate price
relations follow, and the additional error in the revenue identity
is $O_a(K^2\eps+\mathcal W)$. Away from the required endpoint,
nearly optimal floor or upper buyers cannot supply the spending
needed to explain the residual. This preserves the threshold
response. At $\eps=\eps_*(a)/N$, both error terms are small
for sufficiently small $\eps_*(a)$, since $N=\Omega_a(K^2)$.
Appendix~\ref{app:robustness} gives the quantitative estimates.

\paragraph{Constant error and local decoding.}
At constant accuracy, total unapproved expenditure can grow with
the graph. We instead bound the amount received by each vertex
gadget. A vertex decodes correctly whenever this expenditure is
small throughout its radius-two neighborhood: price comparisons
use incoming neighbors, and wire accounting uses both endpoints.
Bounded degree limits the number of vertices affected by any one
gadget receiving too much expenditure.
For errors below a suitable constant multiple of $1/K^2$,
at most $O_a(N\eps)$ constraints fail. For an $n$-vertex source
graph, $N=\Theta_a(K^2n)$, so choosing a sufficiently small
constant accuracy makes the failure fraction arbitrarily small.
Under PCP for PPAD, the required weak threshold problem is hard:
start from weak Pure-Circuit instances in which every node is
used as an input exactly once~\cite{pcpmarkets}, then apply the
threshold gadgets of~\cite{purecircuit}. This preserves bounded
degree and transfers failures with only a constant loss.

\paragraph{The two-slope algorithm.}
On unit bundles, the lower-slope contribution is constant.
Subtracting it and scaling reduces each nonindifferent buyer's
utility to
$U_i(x_i)=\sum_j\min\{x_{ij},\tau_{ij}\}$, preserving demand
and minimum-expenditure tie-breaking; indifferent buyers have
zero thresholds. All quantities below a threshold now have the
same marginal value. Fractional thresholds introduce a difficulty
absent from the bi-valued case: a buyer can exhaust a cheap
good's useful quantity while retaining both money and room in
her bundle.

We extend the increasing-price method of~\cite{vy}, tracking
each buyer's remaining budget and quantity separately.
An initial quantity flow identifies goods that can be assigned
price zero. On these goods, each buyer receives either a full
bundle of useful quantities or her entire threshold quantity
of every identified good, so any unused supply can later fill
bundle deficits without changing utility or expenditure.
At each subsequent stage, we use reciprocal price as a parameter
and apply the Gallo--Grigoriadis--Tarjan parametric maximum-flow
algorithm~\cite{ggt} to find the largest feasible common price
for the remaining goods. A residual-network certificate in the
expenditure network identifies which goods can be committed
at that price and establishes the key invariant:
\emph{every continuing buyer has filled every threshold at
earlier prices}. This ensures that buyers receive cheaper useful
quantities first. Buyers stop when their budget or unit-quantity
constraint binds, or when all their useful quantities have been
allocated. Unused free goods complete their bundles without
changing utility or expenditure. These properties establish
demand optimality and minimum expenditure. Each stage commits
at least one good, giving at most $N$ stages; together with
polynomial bounds on intermediate encoding lengths, the
parametric-flow implementation yields a strongly polynomial
algorithm. Section~\ref{sec:two-slopes} gives the implementation
and proof. With three common slopes, linear utilities $0$,
$at$, and $t$ are already representable by allowing zero-length
segments, so the hardness result follows directly.

\subsection{Related work}

Exact HZ equilibria need not be rational. Vazirani and
Yannakakis~\cite{vy} established this fact and placed their
computation in FIXP. Alaei, Jalaly Khalilabadi, and
Tardos~\cite{ajt} gave exact polynomial-time algorithms, using
algebraic representations, when the number of buyer preference
types or good types is fixed. Garg, Tao, and V\'egh~\cite{gtv}
studied constrained piecewise linear concave utilities, including
matching markets, and obtained approximation schemes when the
number of buyers or goods is fixed. Their existence theorem also
provides normalized equilibria with minimum-expenditure
tie-breaking.

For approximation hardness, Braverman, Liu, Xue, and
Zhou~\cite{blxz} proved a constant-error lower bound under PCP
for PPAD and an unconditional result when buyers cannot spend on
zero-utility goods. Our results permit such expenditure and
restrict every utility to $\{0,a,1\}$. Deligkas, Fearnley,
Hollender, and Melissourgos~\cite{pcpmarkets} proved a related
PCP equivalence for Fisher markets with approximately optimal
bundles. We use their equivalence between weak circuit problems
to obtain the input-once weak Pure-Circuit source for our
reduction~\cite[Theorem~2.2]{pcpmarkets}.

Other work considers allocation rules that obtain fairness and
efficiency through Nash bargaining. Hosseini and
Vazirani~\cite{HV} formulated convex programs for one-sided and
two-sided matching markets with linear and concave utilities.
Tr\"obst and Vazirani~\cite{TV} proved that finding an exactly
envy-free and Pareto-optimal lottery is PPAD-complete, and that
the one-sided Nash-bargaining solution is envy-free within a
factor of two. Combined with the algorithms of Panageas,
Tr\"obst, and Vazirani~\cite{ptv}, this gives polynomial-time
computation of a $(2+\eps)$-approximately envy-free and exactly
Pareto-optimal lottery. These guarantees approximate envy; our
equilibrium notion instead bounds budget violation and utility
loss at the reported prices.

Extensions of HZ also accommodate priorities and allocation
constraints. He, Miralles, Pycia, and Yan~\cite{hmpy} use
priority-dependent prices and establish asymptotic incentive
compatibility. Echenique, Miralles, and Zhang~\cite{emz} allow
allocation constraints arising in settings such as school choice
and course scheduling. The concave model studied here keeps the
unit budgets, supplies, and bundle requirements of HZ. Its
buyer--good breakpoints become capacities in our extension of the
increasing-price method of~\cite{vy}.

\paragraph{Organization and notation.}
Section~\ref{sec:prelim} gives the preliminaries.
Sections~\ref{sec:construction} and~\ref{sec:exact} give the
market construction and prove inverse-linear approximation
hardness. Section~\ref{sec:pcp-equivalence} proves constant-error
hardness under PCP for PPAD. Section~\ref{sec:two-slopes} gives
the two-slope algorithm and the three-slope hardness corollary.
Parameter calculations and quantitative robustness estimates
appear in Appendices~\ref{sec:parameter-checks}
and~\ref{app:robustness}, respectively.
Appendix~\ref{app:rational-witness-proof} contains the
rational-witness proof.

Throughout, $N$ is the number of buyers and unit-supply goods in
a market, and $n=|V|$ is the number of vertices in a threshold
graph. We fix a rational middle utility $a\in(0,1)$ and the
threshold tolerance $\kappa=1/200$. We write
$\clip(t)=\min\{1,\max\{0,t\}\}$ and $\lceil t\rceil$
for the ceiling of $t$. Constants in $O_a(\cdot)$ and
$\Theta_a(\cdot)$ depend only on $a$. Every construction
parameter is fixed before the input graph is given.
\section{Preliminaries}\label{sec:prelim}

\subsection{The source problem: threshold games}

The threshold game of Papadimitriou and Peng~\cite{thresholdgames}
is specified by a directed graph $G=(V,E)$. Write $u\to v$
for $(u,v)\in E$, and let $h_v$ and $o_v$ denote the
indegree and outdegree of $v$.

\begin{definition}[Threshold solution]\label{def:source}
For $0<\kappa<1/2$, a $\kappa$-threshold solution of $G$
is a vector $z\in[0,1]^V$ satisfying, at every vertex $v$,
\[
\begin{aligned}
 \sum_{u:u\to v}z_u>\frac12+\kappa
 &\quad\Longrightarrow\quad z_v\le\kappa,\\
 \sum_{u:u\to v}z_u<\frac12-\kappa
 &\quad\Longrightarrow\quad z_v\ge1-\kappa.
\end{aligned}
\]
If the incoming sum belongs to $[1/2-\kappa,1/2+\kappa]$,
any value $z_v\in[0,1]$ is allowed.
\end{definition}

Such a solution always exists. Finding one is PPAD-complete for
every fixed $0<\kappa<1/6$, even when both degrees are at most
two~\cite[Theorem~4.23]{purecircuit}. We fix $\kappa=1/200$.

We also require positive outdegree. This restriction preserves
hardness: for each sink $v$, introduce fresh vertices $s,t$
and edges $v\to s$, $s\to t$, and $t\to s$. All degrees
remain at most two. Since no original incoming neighborhood
changes, every solution of the augmented graph restricts to a
solution of the original graph. Conversely, every original
solution extends by setting $z_s=0$ and $z_t=1$. Hence we
may assume throughout that
\begin{equation}\label{eq:degrees}
 h_v\in\{0,1,2\},\qquad o_v\in\{1,2\}.
\end{equation}

\subsection{HZ equilibrium and additive approximation}

An HZ market has $N\ge1$ buyers and $N$ divisible goods, with unit
budgets and unit supplies. Every buyer must receive a bundle of
total quantity one. Buyer $i$'s utility for a bundle
$x_i=(x_{ij})_{j=1}^N$ is
\[
 U_i(x_i)=\sum_{j=1}^N u_{ij}x_{ij},
\]
where the nonnegative utilities $u_{ij}$ are rational numbers
given in binary. An allocation $x\ge0$ is \emph{feasible} if
\[
 \sum_jx_{ij}=1\quad\text{for every buyer }i,\qquad
 \sum_ix_{ij}=1\quad\text{for every good }j.
\]
A nonnegative price vector $p$ is \emph{normalized} if
$\min_jp_j=0$. Write $p\cdot x_i=\sum_jp_jx_{ij}$ for
expenditure. At normalized prices, the budget-one optimal utility
is the value of
\begin{equation}\label{eq:value}
\begin{aligned}
 V_i(p)=\max_y\; &\sum_j u_{ij}y_j\\
 \text{subject to}\; &\sum_j y_j=1,\\
 &\sum_j p_jy_j\le1,\\
 &y_j\ge0 &&\text{for every good }j.
\end{aligned}
\tag{OptBundle}
\end{equation}
A normalized exact equilibrium is a pair $(p,x)$ with normalized
prices and feasible allocation such that $p\cdot x_i\le1$
and $U_i(x_i)=V_i(p)$ for every buyer $i$.

For approximation, assume $u_{ij}\in[0,1]$ and use the following
notion from~\cite[Definition~2.3, full version]{ccpy}.

\begin{definition}[Approximate equilibrium]\label{def:approx}
For $\eps\ge0$, a normalized additive $\eps$-equilibrium is
a pair $p,x\ge0$ with $\min_jp_j=0$ and feasible $x$ such
that, for every buyer $i$,
\[
 p\cdot x_i\le1+\eps,\qquad U_i(x_i)\ge V_i(p)-\eps.
\]
The case $\eps=0$ is an exact equilibrium.
\end{definition}

The benchmark $V_i(p)$ uses budget one, including for approximate
equilibria.
An exact equilibrium has \emph{minimum-expenditure tie-breaking}
if every buyer receives a cheapest utility-maximizing affordable
unit bundle. A normalized equilibrium with this property exists
for our linear markets~\cite[Theorem~5.1, conference version]{gtv}.
The hardness results use only Definition~\ref{def:approx}; the
two-slope algorithm also satisfies minimum-expenditure tie-breaking.

\subsection{Groups of goods and optimal demand}

We use \emph{groups} of distinct unit-supply goods with identical
utility vectors: for any two goods $j,k$ in a group,
$u_{ij}=u_{ik}$ for every buyer $i$. Assigning utility $c$ to
a group means assigning utility $c$ to each good in it. All groups
in the construction are disjoint. A buyer \emph{approves} a good
if her utility for it is positive. Her quantity and expenditure
on a group are the sums over its goods.

\begin{lemma}[Equalizing prices within a group]\label{lem:copies}
Replacing all prices in a group by their minimum, while keeping
the allocation fixed, preserves a normalized additive
$\eps$-equilibrium with the same error and the same group minima.
\end{lemma}
\begin{proof}
Allocated bundles become weakly cheaper. A bundle affordable at
the new prices can be purchased at the old prices by moving its
quantity from the group to an originally cheapest good in that
group. The utility and cost are unchanged, and individual demand
has no supply constraints. Conversely, every bundle affordable
at the old prices remains affordable. Each buyer's optimal value
is therefore unchanged, as are allocation feasibility and price
normalization.
\end{proof}

The lemma allows us to use one price per group. To describe
individual demand, we also use \emph{free filler}: an auxiliary
good of price and utility zero that supplies any unused bundle
quantity. It has no supply constraint in the individual problem.

\begin{lemma}[Demand dual and filler]\label{lem:dual}
Let prices be nonnegative with at least one zero-price good,
and let a buyer's utilities lie in $[0,1]$, with
$\max_j u_{ij}=1$. Replacing her unapproved goods by free
filler preserves her optimal utility. The dual of~\eqref{eq:value}
is
\begin{equation}\label{eq:dual}
\begin{aligned}
 V_i(p)=\min_{\beta_i,\mu_i}\; &\beta_i+\mu_i\\
 \text{subject to}\; &\beta_i p_j+\mu_i\ge u_{ij}
                         &&\text{for every }j,\\
 &\beta_i\ge0,\quad\mu_i\in\mathbb R.
\end{aligned}
\tag{Dual}
\end{equation}
Every optimal dual satisfies $0\le\beta_i,\mu_i\le1$.
Every optimal bundle has zero expenditure on unapproved goods.
If every good valued at one costs more than one, every optimal
bundle exhausts the budget.
\end{lemma}
\begin{proof}
Replacing unapproved quantities by filler preserves utility and
weakly reduces cost. Conversely, replacing filler by an actual
zero-price good weakly increases utility at the same cost. Hence
the two problems have the same optimal value.

Linear-programming duality gives~\eqref{eq:dual}, with budget
multiplier $\beta_i$ and quantity multiplier $\mu_i$.
The latter is unrestricted, but the
constraint for a zero-price good implies $\mu_i\ge0$. Strong
duality gives $\beta_i+\mu_i=V_i(p)\le1$, so both multipliers
belong to $[0,1]$. They cannot both be zero because some utility
is one. Consequently, an unapproved good of positive price has
strict dual slack $\beta_i p_j+\mu_i>0$ and cannot occur in
an optimal bundle.

If every good valued at one costs more than one, then
$V_i(p)<1$. A feasible dual with $\beta_i=0$ would require
$\mu_i\ge1$ and could not be optimal. Thus $\beta_i>0$
at every optimum, and complementary slackness forces expenditure
one.
\end{proof}

If all approved goods cost more than one, spending the budget on
goods of maximum utility per dollar uses less than one unit; filler
completes an optimal bundle. A good priced below one can instead
make both constraints bind. We will repeatedly use the following
two-group calculation.

\begin{lemma}[Demand for two groups]\label{lem:demand}
Suppose a buyer values two groups at $a$ and one and all other
goods at zero. Let $q\ge0$ and $p>1$ be their respective
prices. With $x,y$ denoting her quantities from these groups,
the demand problem is
\[
 \max\{ax+y:x,y\ge0,\ x+y\le1,\ qx+py\le1\}.
\]
The optimal pairs are as follows:
\begin{enumerate}
\item If $q>ap$, the unique pair is $(x,y)=(0,1/p)$.
\item If $1\le q<ap$, the unique pair is $(x,y)=(1/q,0)$.
\item If $q<1$ and $q<ap$, both constraints bind and the
      unique pair is
\begin{equation}\label{eq:demand}
 x=\frac{p-1}{p-q},\qquad y=\frac{1-q}{p-q}.
\end{equation}
\item If $q=ap$, every nonnegative pair with $qx+py=1$ and
      $x+y\le1$ is optimal, and these are all the optimal pairs.
\end{enumerate}
In each case, filler supplies quantity $1-x-y$.
\end{lemma}
\begin{proof}
If $q>ap$, the dual $(\beta,\mu)=(1/p,0)$ bounds utility
by $1/p$, attained at $(0,1/p)$. The first group has strict
dual slack, so every optimum has $x=0$. If $1\le q<ap$,
the dual $(a/q,0)$ similarly forces $y=0$ and gives $x=1/q$.

If $q<1$ and $q<ap$, use
\[
 \beta=\frac{1-a}{p-q}>0,\qquad
 \mu=\frac{ap-q}{p-q}>0.
\]
Both group constraints are tight in the dual. Complementary
slackness requires $x+y=1$ and $qx+py=1$. Their solution
is~\eqref{eq:demand}; it is nonnegative and attains the dual
value. Finally, when $q=ap$, utility is $(qx+py)/p$ and is
maximized precisely by feasible pairs that spend one.
\end{proof}

\section{The market construction}\label{sec:construction}

Fix a rational $a\in(0,1)$ and a nonempty directed graph
$G=(V,E)$ satisfying~\eqref{eq:degrees}. We construct an HZ market
whose prices encode a solution of the threshold game on $G$.
Each vertex contributes five groups of goods, and each edge
contributes buyers who approve goods at its two endpoints.
The preferences use only $\{0,a,1\}$. We specify these preferences
first and choose the integer supplies and population sizes in
Section~\ref{sec:exact}.

All goods have unit supply. All buyers have budget one and require
one unit of goods in total. Groups and populations are disjoint,
and every utility not specified below is zero.

\subsection{Goods and preferences}

For each vertex $v$, introduce target groups $L_v,M_v,H_v$ with
supplies $S_L,S_M,S_H$, and reference groups
$\mathcal R_{Fv},\mathcal R_{Uv}$ with supplies
$S^R_{Fv},S^R_{Uv}$. The low and middle supplies depend on $o_v$;
we omit that argument from the notation. The high supply is the
same at every vertex.

The vertex has two populations linking consecutive targets.
The $b_v$ lower-link buyers $\mathcal B_v^{LM}$ assign utility
$a$ to $L_v$ and utility one to $M_v$. The $c_v$ upper buyers
$\mathcal B_v^{MH}$ assign utility $a$ to $M_v$ and
$\mathcal R_{Uv}$, and utility one to $H_v$.
We will choose $b_v$ as a function of $h_v$ and $c_v$ as a
function of $o_v$.

There are also $F_v$ floor buyers $\mathcal F_v$, who assign
utility $a$ to $L_v$ and utility one to $M_v$ and
$\mathcal R_{Fv}$. Each reference group has its own anchors:
the $A^R_{Fv}$ buyers $\mathcal A_{Fv}$ approve only
$\mathcal R_{Fv}$, and the $A^R_{Uv}$ buyers $\mathcal A_{Uv}$
approve only $\mathcal R_{Uv}$, in each case at utility one.
Thus the floor reference is approved only by floors and its
anchors, and the upper reference only by upper buyers and its
anchors. Figure~\ref{fig:vertex} records all vertex preferences.

The global group $\mathcal Z$ consists of goods valued at zero
by every buyer. Its size will make the total supply equal to the
number of buyers.

\begin{figure}[htbp]
\centering
\begin{tikzpicture}[x=1cm,y=1cm]
\node[hzgood] (low) at (-5,0) {$L_v$};
\node[hzgood] (mid) at (0,0) {$M_v$};
\node[hzgood] (high) at (5,0) {$H_v$};
\node[hzbuyer] (lower) at (-2.5,0) {$\mathcal B_v^{LM}$};
\node[hzbuyer] (upper) at (2.5,0) {$\mathcal B_v^{MH}$};
\draw[hzedge] (low)--node[hzcoef,above]{$a$}(lower);
\draw[hzedge] (lower)--node[hzcoef,above]{$1$}(mid);
\draw[hzedge] (mid)--node[hzcoef,above]{$a$}(upper);
\draw[hzedge] (upper)--node[hzcoef,above]{$1$}(high);
\node[hzbuyer] (floor) at (-2.5,1.6) {$\mathcal F_v$};
\node[hzgood] (rf) at (0,2.8) {$\mathcal R_{Fv}$};
\node[hzbuyer] (af) at (0,4.2) {$\mathcal A_{Fv}$};
\draw[hzedge] (low)--node[hzcoef,above left]{$a$}(floor);
\draw[hzedge] (floor)--node[hzcoef,above right]{$1$}(mid);
\draw[hzedge] (floor)--node[hzcoef,above]{$1$}(rf);
\draw[hzedge] (rf)--node[hzcoef,right]{$1$}(af);
\node[hzgood] (ru) at (2.5,-2.1) {$\mathcal R_{Uv}$};
\node[hzbuyer] (au) at (-1,-2.1) {$\mathcal A_{Uv}$};
\node[hzgood,draw=black!55,fill=black!3] (zero) at (-5,-2.1) {$\mathcal Z$};
\draw[hzedge] (upper)--node[hzcoef,right]{$a$}(ru);
\draw[hzedge] (ru)--node[hzcoef,above]{$1$}(au);
\end{tikzpicture}
\caption{The vertex gadget at $v$. Solid blue ellipses denote groups of
 goods; dashed orange ellipses denote buyer populations. Each edge label
 gives every buyer's utility for each good in the adjacent group.
The balancing group $\mathcal Z$ is shared by the whole market and
has no positive-utility connections.}
\label{fig:vertex}
\end{figure}
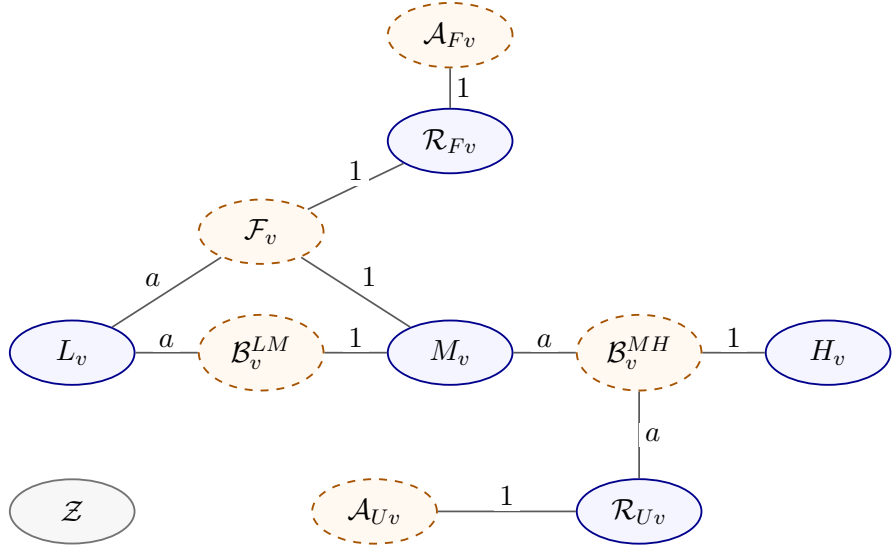

\FloatBarrier

For each directed edge $u\to v$, introduce a population
$\mathcal E_{uv}$ of $K$ wire buyers, each assigning utility $a$
to $L_u$ and utility one to $H_v$. Consequently, $L_v$ has
$o_vK$ outgoing wire approvers and $H_v$ has $h_vK$ incoming
wire approvers. A self-loop contributes one population, approving
$L_v$ and $H_v$; it is counted only once. See
Figure~\ref{fig:edges}.

\begin{figure}[htbp]
\centering
\begin{tikzpicture}
\node[hzgood] (src) at (-4,0) {$L_u$};
\node[hzbuyer] (wire) at (0,0) {$\mathcal E_{uv}$};
\node[hzgood] (dst) at (4,0) {$H_v$};
\draw[hzedge] (src)--node[hzcoef,above]{$a$}(wire);
\draw[hzedge] (wire)--node[hzcoef,above]{$1$}(dst);
\end{tikzpicture}
\caption{The wire gadget for $u\to v$: $K$ buyers connect $L_u$ to $H_v$.
For a self-loop $u=v$, this is one population connecting $L_v$ and $H_v$.}
\label{fig:edges}
\end{figure}
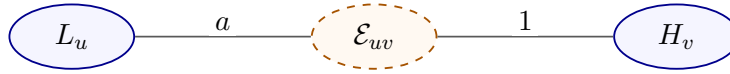

\FloatBarrier

Let $N$ be the total number of buyers and $G_{\rm listed}$ the
number of goods in the five groups per vertex. Add
$N-G_{\rm listed}$ goods to $\mathcal Z$.
Lemma~\ref{lem:balancing} proves that this number is positive for
our parameter choice. The resulting market has $N$ buyers and
$N$ goods. Table~\ref{tab:vertex-populations} lists all groups
and populations.

\begin{table}[htbp]
\centering
\caption{The groups of goods and buyer populations. Each listed utility
applies to every good in the indicated group. All unlisted utilities
are zero.}
\label{tab:vertex-populations}
\renewcommand{\arraystretch}{1.2}
\begin{tabularx}{\textwidth}{|>{\centering\arraybackslash}p{.30\textwidth}|>{\centering\arraybackslash}p{.18\textwidth}|>{\centering\arraybackslash}X|}
\hline
\rowcolor{blue!10}
Good group & Number & Role\\\hline
$L_v$ & $S_L$ & Low target\\\hline
$M_v$ & $S_M$ & Middle target\\\hline
$H_v$ & $S_H$ & High target\\\hline
$\mathcal R_{Fv}$ & $S^R_{Fv}$ & Floor reference\\\hline
$\mathcal R_{Uv}$ & $S^R_{Uv}$ & Upper reference\\\hline
$\mathcal Z$ & $N-G_{\rm listed}$ & Balancing goods; utility zero for every buyer\\\hline
\rowcolor{blue!10}
Buyer population & Number & Nonzero utilities\\\hline
Lower links $\mathcal B_v^{LM}$ & $b_v$ & $a$ on $L_v$; $1$ on $M_v$\\\hline
Upper buyers $\mathcal B_v^{MH}$ & $c_v$ & $a$ on $M_v,\mathcal R_{Uv}$; $1$ on $H_v$\\\hline
Floors $\mathcal F_v$ & $F_v$ & $a$ on $L_v$; $1$ on $M_v,\mathcal R_{Fv}$\\\hline
Floor anchors $\mathcal A_{Fv}$ & $A^R_{Fv}$ & $1$ on $\mathcal R_{Fv}$\\\hline
Upper anchors $\mathcal A_{Uv}$ & $A^R_{Uv}$ & $1$ on $\mathcal R_{Uv}$\\\hline
Wires $\mathcal E_{uv}$, $u\to v$ & $K$ & $a$ on $L_u$; $1$ on $H_v$\\\hline
\end{tabularx}
\end{table}

\FloatBarrier

\subsection{Prices and the decoded variables}

For a price vector, write
\[
 \begin{aligned}
 q_v&=\min_{j\in L_v}p_j,&
 m_v&=\min_{j\in M_v}p_j,&
 p_v&=\min_{j\in H_v}p_j,\\
 r_{Fv}&=\min_{j\in\mathcal R_{Fv}}p_j,&
 r_{Uv}&=\min_{j\in\mathcal R_{Uv}}p_j.
 \end{aligned}
\]
Lemma~\ref{lem:copies} allows us to replace every price in a
group by its minimum, preserving the equilibrium error. We use
these equalized prices in the analysis; the decoder depends only
on the group minima and is therefore unchanged.

The high price $p_v$ represents the variable at $v$.
For endpoints $L<U$, to be chosen in~\eqref{eq:interval}, put
$w=U-L$ and define
\begin{equation}\label{eq:decoder}
 t_v=\frac{p_v-L}{w},\qquad z_v=\clip(t_v).
\end{equation}
A positive rational parameter $\tau\le\kappa/100$ determines
$\eta=\tau w$. We choose the two reference groups so that every
exact equilibrium satisfies
\[
 aL\le r_{Fv}\le a(L+\eta),\qquad
 aU\le r_{Uv}\le a(U+\eta).
\]
These bounds will imply $L\le p_v\le U+\eta$.

Denote the completed market by $\mathcal M(G;a,\tau,K)$, with
all supplies and counts given in Subsection~\ref{sec:choose-K}.
For fixed $a,\tau,K$, each count depends only on the local degrees
and is bounded independently of the graph. The hardness results use one sufficiently small rational value of
$\tau$, denoted $\tau_a$, that depends only on $a$; its existence
is established in Theorem~\ref{thm:robust-decoding}.

\section{Hardness of inverse-linear approximation}\label{sec:exact}

\begin{theorem}[Soundness of the market construction]
\label{thm:soundness}\label{thm:encoding}
For every fixed rational $a\in(0,1)$, there are positive rational
constants $\tau_a,\eps_*(a)$ and an integer $K_a$, depending only
on $a$, with the following property. For every $K\ge K_a$
admissible as in Subsection~\ref{sec:choose-K} and every nonempty
graph $G$ satisfying~\eqref{eq:degrees},
every normalized additive $\eps_*(a)/N$-equilibrium of
$\mathcal M(G;a,\tau_a,K)$ decodes to values satisfying,
at every vertex,
\[
 \begin{aligned}
 \sum_{u\to v}z_u>\tfrac12+\kappa
 &\quad\Longrightarrow\quad z_v\le\kappa,\\
 \sum_{u\to v}z_u<\tfrac12-\kappa
 &\quad\Longrightarrow\quad z_v\ge1-\kappa.
 \end{aligned}
\]
Thus the decoded vector is a solution of the threshold game.
\end{theorem}

The exact analysis has two parts. First, market clearing forces
$q_v=a^2p_v$ and $m_v=ap_v$. Second, the resulting wire payments
relate the incoming sum to the expenditure of floor and upper
buyers. When the incoming sum is outside the threshold interval,
this forces $p_v$ to the corresponding lower or upper endpoint.
We derive the necessary size inequalities
in this order. All supplies and buyer counts are positive integers,
and each requirement remains in force after it is introduced.
Subsection~\ref{sec:choose-K} gives a simultaneous choice; the
arithmetic verification is in Appendix~\ref{sec:parameter-checks}.
The subsequent robustness argument proves the theorem from the
same expenditure identities.

\subsection{Reference prices and lower bounds on target prices}
\label{sec:design}\label{sec:prices}

Fix a normalized exact equilibrium and equalize its group prices
using Lemma~\ref{lem:copies}. By Lemma~\ref{lem:dual}, no buyer
spends on goods she values at zero. The following observation
controls each reference price independently of the other groups.

\begin{lemma}[Reference price bounds]\label{lem:generic-reference}
Consider a group of supply $S$ with $A>S$ anchors, each assigning
utility one to this group and zero to every other good. Suppose
at most $T$ other buyers approve the group, and let $E$ be their
total expenditure on it. At an exact equilibrium, its price $r$
satisfies
\[
 rS=A+E,\qquad 0\le E\le T,
 \qquad\text{and hence}\qquad
 \frac AS\le r\le\frac{A+T}{S}.
\]
\end{lemma}
\begin{proof}
If $r\le1$, every anchor has optimal utility one and must receive
one unit from the group. This is impossible because $A>S$.
Thus $r>1$. Each anchor buys $1/r$ units and spends one, while
all other approvers together spend between zero and $T$.
Market clearing gives the identity.
\end{proof}

For $\lambda>1$ and $\eta>0$, the choices
\[
 S=\left\lceil\frac{T+1}{a\eta}\right\rceil,
 \qquad A=\lceil\lambda S\rceil
\]
give $\lambda\le r\le\lambda+(T+1)/S\le\lambda+a\eta$.
Apply this construction with $(\lambda,T)=(aL,F_v)$ and
$(aU,c_v)$, respectively:
\begin{equation}\label{eq:reference-counts}
 \begin{aligned}
 S^R_{Fv}&=\left\lceil\frac{F_v+1}{a\eta}\right\rceil,
 &A^R_{Fv}&=\left\lceil aL S^R_{Fv}\right\rceil,\\
 S^R_{Uv}&=\left\lceil\frac{c_v+1}{a\eta}\right\rceil,
 &A^R_{Uv}&=\left\lceil aU S^R_{Uv}\right\rceil.
 \end{aligned}
\end{equation}
We impose
\begin{equation}\label{eq:reference-requirements}
 aL\ge\frac{1+a}{2a},\qquad
 a(U+\eta)<\frac1a,\qquad L+\eta<U.
\end{equation}
The two reference prices then satisfy
\[
 aL\le r_{Fv}\le a(L+\eta)<aU\le r_{Uv}\le a(U+\eta).
\]
In particular, both references cost more than one, and a floor
buyer can obtain utility $1/r_{Fv}>a$ from her reference alone.

Choose the floor population large enough that the targets cannot
supply a fixed positive quantity to every floor buyer:
\begin{equation}\label{eq:floor-count}
 \theta=\min\left\{a,\frac{1-a}{2}\right\},\qquad
 F_v=\left\lceil\frac{S_L+S_M}{\theta}\right\rceil+1.
\end{equation}

\begin{lemma}[Reference intervals and target floors]
\label{lem:references}\label{lem:floors}
Under~\eqref{eq:reference-counts}, \eqref{eq:reference-requirements},
and~\eqref{eq:floor-count}, the
reference prices lie in the intervals above and
\[
 q_v\ge ar_{Fv},\qquad m_v\ge r_{Fv}.
\]
A floor buyer spends on low goods only if $q_v=ar_{Fv}$, and on
middle goods only if $m_v=r_{Fv}$.
\end{lemma}
\begin{proof}
The reference bounds follow from Lemma~\ref{lem:generic-reference}.
Write $r=r_{Fv}$ and suppose $q_v<ar$ or $m_v<r$.
One target is then free or has a strictly greater utility-per-dollar
ratio than the reference. An optimal floor bundle cannot contain
both reference goods and zero-utility filler. Indeed, transferring
a small amount of reference expenditure to the better target and
adjusting the filler improves utility; a free target can directly
replace filler.

If the bundle contains no reference goods, its target quantity
exceeds $a$, since its utility is at least $1/r>a$ and every
unit has utility at most one. If it contains no filler, its
reference quantity is at most $1/r$, and its target quantity is
at least
\[
 1-\frac1r\ge\frac{1-a}{1+a}\ge\frac{1-a}{2}.
\]
Thus every floor buyer receives at least $\theta$ target units,
contrary to $F_v\theta>S_L+S_M$.

It follows that the reference has a weakly greater
utility-per-dollar ratio than either target. The bundle consisting
of $1/r$ reference units and free filler attains utility $1/r$.
Any positive target expenditure in an optimal bundle therefore
requires equality of that target's ratio with the reference ratio.
\end{proof}

\subsection{The target-price equalities}

To prove $m_v=ap_v$, we first show that upper buyers must spend
on high goods. This requires a bound on incoming wire expenditure
that does not assume either target-price equality.
Set $\gamma=a^2/(1-a^2)$ and choose
\begin{equation}\label{eq:interval}
 w=U-L=\frac{2}{\gamma K},\qquad
 L=\frac1{a^2}-\frac{3}{\gamma K},\qquad
 U=\frac1{a^2}-\frac{1}{\gamma K},
\end{equation}
\begin{equation}\label{pcp:reference-tolerance}
 \eta=\tau w,\qquad 0<\tau\le\kappa/100.
\end{equation}
Thus $1-a^2L=3(1-a^2)/K$. Together with the price bounds below,
this identity bounds each incoming wire buyer's high expenditure
by $O_a(1/K)$.
Assume $K\ge48$ and require
\begin{equation*}
 (L+\eta)S_H<c_v.
\end{equation*}
For a constant $C>0$, set
\begin{equation}\label{eq:low-supply-choice}
 S_L=o_vK+C,
\end{equation}
with $C<b_v$ and
\begin{equation*}
 \frac{o_vK+b_v}{S_L}<aL.
\end{equation*}
The strict excess $S_L>o_vK$ will also be needed for the lower link.

\begin{lemma}[Preliminary price and wire bounds]
\label{lem:high-floor}\label{lem:coarse}
Under these requirements,
\[
 p_v\ge r_{Fv}/a\ge L,\qquad
 a^2L\le q_v\le\frac{o_vK+b_v}{S_L}<aL.
\]
All incoming wires together spend less than $9$ on $H_v$.
\end{lemma}
\begin{proof}
Suppose $p_v<r_{Fv}/a$. The floor bounds and reference intervals
give $m_v>ap_v$ and $r_{Uv}>ap_v$. Both alternatives cost
more than one. Each upper buyer therefore demands
$1/\max\{1,p_v\}$ high units. Since $p_v<L+\eta$ and
$L+\eta>1$, their total demand exceeds
$c_v/(L+\eta)>S_H$, a contradiction.

The floor bound gives $q_v\ge ar_{Fv}\ge a^2L$.
If floors spend on low goods, then $q_v=ar_{Fv}<1$.
Otherwise the entire low revenue comes from lower links and
outgoing wires, so $q_vS_L\le b_v+o_vK$. Because $C<b_v$,
$(o_vK+b_v)/S_L>1$, and the stated upper bound holds in both cases.

An incoming wire faces a low price $q<aL\le ap$ and a high
price $p\ge L$. If $q\ge1$, it buys no high goods.
If $q<1$, Lemma~\ref{lem:demand} gives high expenditure
$p(1-q)/(p-q)$. Using $q\ge a^2L$, $q<1$, and $p\ge L$,
\[
 \frac{p(1-q)}{p-q}
 \le\frac{L}{L-1}(1-a^2L)
 =\frac{3a^2wL}{2(L-1)}
 <\frac{2a^2wL}{L-1}.
\]
There are at most $2K$ incoming buyers, whose total high
expenditure is therefore at most
\[
 \frac{4Ka^2wL}{L-1}
 =8+\frac{24a^2}{K-3}<9.
\]
\end{proof}

We now impose the remaining inequalities used to prove the links:
\begin{equation}\label{eq:link-requirements}
 LS_H>10,\qquad aLS_M>c_v+10,\qquad c_v>b_v,\qquad
 \frac{3b_v(1+a)}K<1.
\end{equation}
The first requires more high revenue than incoming wires can
provide. The second requires more middle revenue than upper
buyers can provide alone. These two facts rule out the strict
price comparisons in the next two lemmas.

\begin{lemma}[The upper link]\label{lem:upper-link}
The middle and high prices satisfy
\[
 m_v=ap_v,\qquad p_v\le r_{Uv}/a.
\]
\end{lemma}
\begin{proof}
High revenue is at least $LS_H>10$, and incoming wires contribute
less than $9$. Hence some upper buyer purchases high goods.
All her approved goods cost more than one, so she spends only
on groups maximizing utility per dollar. Her positive high
expenditure implies $m_v\ge ap_v$ and $r_{Uv}\ge ap_v$.

If $m_v>ap_v$, upper buyers purchase no middle goods.
Floor buyers purchase none either, because
$m_v>ap_v\ge r_{Fv}$. Lower links would then have to provide
all middle revenue. Their total budget $b_v$ is smaller than
$m_vS_M\ge aLS_M>c_v>b_v$, a contradiction.
Thus $m_v=ap_v$.
\end{proof}

\begin{lemma}[The lower link]\label{lem:lower-link}
The low price satisfies $q_v=am_v=a^2p_v<1$.
\end{lemma}
\begin{proof}
Suppose first that $q_v>am_v$. Lower links then purchase no
low goods. Since $m_v\ge r_{Fv}$, floor buyers also purchase
none. The $o_vK$ outgoing wire buyers cannot clear
$S_L>o_vK$ low goods because each receives only one unit in total.
Consequently,
\[
 q_v\le am_v=a^2p_v\le a^2(U+\eta)<1.
\]

Suppose instead that $q_v<am_v$. Since $q_v\ge ar_{Fv}$,
we have $m_v>r_{Fv}$, and floor buyers purchase no middle
goods. Lemma~\ref{lem:demand} bounds each lower buyer's middle
expenditure by
\[
 \frac{m_v(1-q_v)}{m_v-q_v}
 <\frac{1-q_v}{1-a}
 \le\frac{1-a^2L}{1-a}
 =\frac{3(1+a)}K.
\]
By~\eqref{eq:link-requirements}, lower links contribute less
than one in total, and upper buyers contribute at most $c_v$.
This is insufficient for the middle revenue
$m_vS_M\ge aLS_M>c_v+10$. Hence $q_v=am_v$.
\end{proof}

\begin{proposition}[The price chain]\label{prop:chain}
Every normalized exact equilibrium has equal prices within each
group and satisfies
\[
 q_v=am_v=a^2p_v,\qquad q_v<1<m_v<p_v,\qquad
 \frac{r_{Fv}}a\le p_v\le\frac{r_{Uv}}a.
\]
Every buyer spends one. Lower links spend only on the targets.
Upper buyers spend only on the targets when $p_v<r_{Uv}/a$;
floor buyers spend on targets only when $p_v=r_{Fv}/a$.
\end{proposition}
\begin{proof}
The preceding lemmas give all the price relations for the
equalized equilibrium. Every group assigned utility one has price
above one. Lemma~\ref{lem:dual} therefore implies that every
buyer spends her full budget.

Lower links approve only targets. If $p_v<r_{Uv}/a$, the upper
reference has a strictly smaller utility-per-dollar ratio than
either upper target, so upper buyers spend only on the targets.
At $p_v=r_{Uv}/a$, all their approved groups have the same
ratio; any division of the unit budget among them is optimal and
can be completed with free filler. For floor buyers, the target
spending equalities in Lemma~\ref{lem:floors} are both equivalent
to $p_v=r_{Fv}/a$.

Finally, the equalized prices have total revenue $N$. The original
prices are no smaller, while their total revenue is at most $N$.
Thus no price decreased under equalization, and the original
prices were already equal within each group.
\end{proof}

\subsection{Expenditure identities and exact decoding}\label{sec:clearing}

The price chain gives an explicit expression for each wire payment.
For $u\to v$, we have $q_u=a^2p_u<1<p_v$ and $q_u<ap_v$.
By Lemma~\ref{lem:demand}, every wire buyer receives the following
quantities from $L_u$ and $H_v$, respectively:
\begin{equation}\label{eq:wire-quantities}
 x=\frac{p_v-1}{p_v-q_u},\qquad
 y=\frac{1-q_u}{p_v-q_u}.
\end{equation}
The $K$ buyers together spend $I_{uv}$ on $H_v$ and $K-I_{uv}$
on $L_u$, where
\begin{equation}\label{eq:wire-dollars}
 I_{uv}=\frac{Kp_v(1-a^2p_u)}{p_v-a^2p_u}.
\end{equation}
Recall $t_u=(p_u-L)/w$ and $z_u=\clip(t_u)$.
The reference bounds imply $0\le t_u\le1+\tau$ and
$0\le t_u-z_u\le\tau$. Substituting $p_u=L+wt_u$ into
\eqref{eq:wire-dollars}, and using
$K(1-a^2p_u)=(1-a^2)(3-2t_u)$, gives
\begin{equation}\label{eq:wire-expansion}
 I_{uv}=3-2t_u+e_{uv},\qquad
 e_{uv}=\frac{a^2(3-2t_u)(p_u-p_v)}{p_v-a^2p_u}.
\end{equation}
Since $p_u,p_v\in[L,U+\eta]$ and $p_v-a^2p_u\ge L-1$,
\begin{equation*}
 |e_{uv}|\le\frac{3a^2(w+\eta)}{L-1}=O_a(1/K).
\end{equation*}
For a self-loop, $e_{vv}=0$.

At vertex $v$, let $R_v$ be the total revenue of its three
target groups, $E_v$ their receipts from wires, and $W_v,V_v$
the expenditures of floor and upper buyers on these targets.
Define $\Phi_v$ to be target revenue minus wire receipts and the
full budgets of the two link populations. Explicitly,
\[
 \begin{aligned}
 R_v&=q_vS_L+m_vS_M+p_vS_H
      =p_v(a^2S_L+aS_M+S_H),\\
 W_v&\ge0,\qquad 0\le V_v\le c_v,
 \end{aligned}
\]
\begin{equation}\label{eq:E}
 E_v=\sum_{v\to t}(K-I_{vt})+\sum_{u\to v}I_{uv},
\end{equation}
\begin{equation}\label{eq:F}
 \Phi_v:=R_v-E_v-b_v-c_v=W_v+(V_v-c_v).
\end{equation}
The equality in~\eqref{eq:F} is market clearing: the lower links
spend exactly $b_v$ on targets, and the remaining target receipts
are $E_v,W_v,V_v$.

Equation~\eqref{eq:F} relates the sign of $\Phi_v$ to the
location of $p_v$. If $\Phi_v>0$, then $W_v>0$, so floors
purchase targets. If $\Phi_v<0$, then $V_v<c_v$, so upper
buyers spend on their reference. Proposition~\ref{prop:chain}
therefore implies
\[
 \Phi_v>0\quad\Longrightarrow\quad p_v=r_{Fv}/a\le L+\eta,
 \qquad
 \Phi_v<0\quad\Longrightarrow\quad p_v=r_{Uv}/a\ge U.
\]
It remains to express $\Phi_v$ in terms of the incoming variables.
Writing $e_v=\sum_{u\to v}e_{uv}-\sum_{v\to t}e_{vt}$,
\eqref{eq:wire-expansion} yields
\[
 E_v=o_v(K-3+2t_v)+3h_v-2\sum_{u\to v}t_u+e_v.
\]
We require the target revenue and link budgets to satisfy
\[
 a^2S_L+aS_M+S_H=o_v\gamma K,\qquad
 b_v+c_v=o_v\gamma K+1-3h_v.
\]
Because $\gamma Kw=2$ and $\gamma K(L-1)=K-3$, these identities
cancel the term $2o_vt_v$ from $\Phi_v$ and leave the threshold
constant $1/2$ after division by two. Both identities and
\eqref{eq:low-supply-choice} hold if, for a constant $b>2C$, we set
\begin{equation}\label{eq:target-size-choice}
 \begin{aligned}
 b_v&=b+1-3h_v,& c_v&=o_v\gamma K-b,\\
 S_L&=o_vK+C,& S_M&=a(c_v+C),& S_H&=a^2(b-2C).
 \end{aligned}
\end{equation}
Indeed, $R_v=o_v\gamma Kp_v=o_v\gamma K+o_v(K-3+2t_v)$,
and subtraction gives
\begin{equation}\label{eq:canceled}
 \Phi_v=2\sum_{u\to v}t_u-1-e_v.
\end{equation}
A self-loop contributes $t_v$ to this incoming sum, as required
by the threshold game.

There are at most four incident wire errors and two incoming
clipping errors. Thus
\begin{equation}\label{eq:threshold-bound}
 \left|\frac{\Phi_v}{2}
       -\left(\sum_{u\to v}z_u-\frac12\right)\right|
 \le\frac{6a^2w(1+\tau)}{L-1}+2\tau.
\end{equation}
The parameter choice in the next subsection makes the right-hand
side less than $\kappa/10$; see~\eqref{eq:final-checks}.

\begin{proposition}[Exact decoding]\label{prop:exact}
For that parameter choice, every exact equilibrium decodes to
values satisfying
\begin{align*}
 \sum_{u\to v}z_u\ge\frac12+\frac\kappa2
 &\quad\Longrightarrow\quad z_v\le\frac\kappa4,\\
 \sum_{u\to v}z_u\le\frac12-\frac\kappa2
 &\quad\Longrightarrow\quad z_v\ge1-\frac\kappa4.
\end{align*}
\end{proposition}
\begin{proof}
In the first case, \eqref{eq:threshold-bound} gives $\Phi_v>0$.
Hence $p_v\le L+\eta$ and $z_v\le\eta/w=\tau<\kappa/4$.
In the second case, $\Phi_v<0$, so $p_v\ge U$ and $z_v=1$.
\end{proof}

\subsection{Choosing the integer multiplicities}\label{sec:choose-K}

We now choose the remaining constants to satisfy the inequalities
used in the exact analysis. Write $a=A/D$ in lowest terms and put
\[
 C=16D^2,\qquad b=4C.
\]
For any rational $0<\tau\le\kappa/100$, choose a positive
multiple $K$ of $D^2(D^2-A^2)$ such that
\begin{equation}\label{eq:check-basic}
 K\ge\frac{8b+200/\kappa}{a^2(1-a)}.
\end{equation}
We call these values of $K$ \emph{admissible}.
Evaluate the remaining parameters in the following order:
first $L,U,w,\eta$ from~\eqref{eq:interval} and
\eqref{pcp:reference-tolerance}; then $b_v,c_v,S_L,S_M,S_H$
from~\eqref{eq:target-size-choice}; then $F_v$
from~\eqref{eq:floor-count}; and finally the reference supplies
and anchor counts from~\eqref{eq:reference-counts}.

Lemmas~\ref{lem:parameters} and~\ref{lem:clearing-margins} prove
that all counts and supplies are positive integers and that every
inequality used above holds. They also establish the error bound
in~\eqref{eq:final-checks}. Lemma~\ref{lem:balancing} proves
$N>G_{\rm listed}$, completing the construction.

For approximate equilibria, we will choose a fixed value
$\tau_a\le\kappa/100$ and an additional cutoff $K_a$, both
depending only on $a$. Theorem~\ref{thm:robust-decoding} provides
these constants and applies to every admissible $K\ge K_a$.
This uniformity is used again in Section~\ref{sec:pcp-equivalence}.

\subsection{Robustness to approximation}\label{sec:approximation}
\label{sec:local-robustness}\label{sec:robustness}

For inverse-linear accuracy, we make the total expenditure on
unapproved goods sufficiently small and prove correct decoding at
every vertex. Section~\ref{sec:pcp-threshold-to-hz} will use a
local version of this argument when some decoding failures are
permitted.

Put $n=|V|$ and equalize group prices using Lemma~\ref{lem:copies}.
This preserves the equilibrium error and the minimum-price decoder.
Let $\mathcal G_v$ be the union of the five groups at $v$, and put
\[
 \mathcal W=\sum_i\sum_{j:u_{ij}=0}p_jx_{ij},\qquad
 \mathcal W_v=\sum_i\sum_{j\in\mathcal G_v:\,u_{ij}=0}p_jx_{ij}.
\]
Thus $\mathcal W$ is total unapproved expenditure and
$\sum_v\mathcal W_v\le\mathcal W$. The quantity $\mathcal W_v$
is distinct from the floor buyers' target expenditure $W_v$.
All implicit constants below depend only on $a$, uniformly over
the graph, admissible $K$, and $0<\tau\le\kappa/100$.

\begin{lemma}[Unapproved expenditure]\label{lem:pcp-waste}
There are positive rational constants $d_a,g_a,c_a,C_a$, depending
only on $a$, such that for every sufficiently large admissible $K$,
every normalized additive $\eps$-equilibrium with
$\eps\le\min\{c_a,1/(C_aK)\}$ satisfies
\[
 \mathcal W\le C_aN\eps.
\]
Every group valued at one has price greater than $1+d_a$, and
every buyer has optimal utility at most $1-g_a$.
\end{lemma}
\begin{proof}
Lemma~\ref{lem:robust-utility-gap} gives the price and utility
gaps. For each buyer $i$, choose a group $f(i)$ valued at one.
Her optimal demand dual satisfies
\[
 \beta_i(p_{f(i)}-1)\ge1-V_i(p)\ge g_a.
\]
By Lemma~\ref{lem:robust-demand-slack}, her unapproved expenditure
is at most $2\eps/\beta_i\le(2\eps/g_a)p_{f(i)}$.

Choose middle goods for lower links, the floor reference for floor
buyers, each anchor's own reference, and high goods for upper
buyers and wires. For middle and reference groups, the number of
buyers selecting a group is at most a constant times its supply;
for high groups, it is at most $O_a(K)$ times the supply. Since
total revenue is at most $N(1+\eps)$, writing
$R_H=S_H\sum_vp_v$ gives
\[
 \sum_i p_{f(i)}=O_a(N+KR_H),\qquad
 R_H\le\sum_v(c_v+h_vK)(1+\eps)+\mathcal W
       =O_a(Kn)+\mathcal W.
\]
The reference counts in~\eqref{eq:reference-counts} give
$N=\Omega_a(K^2n)$ throughout the permitted range of $\tau$.
Consequently,
\[
 \mathcal W\le O_a(\eps)(N+K\mathcal W).
\]
Taking $C_a$ sufficiently large and $c_a$ sufficiently small,
the assumption $\eps\le\min\{c_a,1/(C_aK)\}$ lets us move
the term proportional to $K\eps\mathcal W$ to the left and
obtain the claimed bound.
\end{proof}

\begin{lemma}[Global decoding]\label{lem:global-decoding}
There are positive rational constants $\tau_a\le\kappa/100$,
$\zeta_a,e_a$, and an integer $K_a\ge1$, depending only on $a$,
such that the following holds for every admissible $K\ge K_a$.
If a normalized additive $\eps$-equilibrium of
$\mathcal M(G;a,\tau_a,K)$ satisfies
\[
 \eps\le e_a/K^2,\qquad \mathcal W\le\zeta_a,
\]
then its decoded vector is a $\kappa$-threshold solution of $G$.
The constants can be chosen so that $\eps\le e_a/K^2$ also
implies the hypotheses of Lemma~\ref{lem:pcp-waste}.
\end{lemma}
\begin{proof}
Since $\mathcal W_v\le\mathcal W$ at every vertex, the price
and expenditure estimates in Appendix~\ref{app:robustness}
apply everywhere when $K$ is large and $\tau_a$, $K^2\eps$,
and $\mathcal W$ are sufficiently small. In particular,
Lemma~\ref{lem:robust-price-chain} gives
\[
 q_v=a^2p_v+O_a(K\eps),\qquad m_v=ap_v+O_a(K\eps).
\]
Evaluate $\Phi_v$ using the ideal prices $a^2p_u,ap_u,p_u$,
as in Lemma~\ref{lem:robust-accounting}. That lemma gives
\[
 \begin{aligned}
 \Phi_v&=W_v+(V_v-c_v)+O_a(K^2\eps+\mathcal W),\\
 \frac{\Phi_v}{2}
 &=\sum_{u\to v}z_u-\frac12+O_a(1/K+\tau_a+K^2\eps).
 \end{aligned}
\]

Suppose $\sum_{u\to v}z_u>1/2+\kappa$ but $z_v>\kappa$.
The second identity gives $\Phi_v>\kappa$ for sufficiently
small errors. Also $p_v>L+\kappa w$, so the approximate price
relations and reference bounds give reduced slack
$\Omega_a(\kappa w)$ on both floor targets. The demand estimate
in Lemma~\ref{lem:robust-demand-slack}, together with
$F_v=O_a(K)$ and $w=\Theta_a(1/K)$, yields
\[
 W_v=O_a(K^2\eps/\kappa),\qquad V_v-c_v\le c_v\eps.
\]
The first identity therefore gives
$\Phi_v\le O_a(K^2\eps/\kappa+\mathcal W)$, a contradiction.

If instead $\sum_{u\to v}z_u<1/2-\kappa$ but
$z_v<1-\kappa$, then $\Phi_v<-\kappa$ and
$p_v<U-\kappa w$. The upper reference has reduced slack
$\Omega_a(\kappa w)$, so upper buyers spend at most
$O_a(K^2\eps/\kappa)$ on it. Their total approved expenditure
is $c_v+O_a(K\eps)$, giving
$V_v-c_v\ge-O_a(K^2\eps/\kappa)$. Since $W_v\ge0$, the
first identity now gives
$\Phi_v\ge-O_a(K^2\eps/\kappa+\mathcal W)$, again a
contradiction.

Both contradictions hold once
\[
 H_a\left(1/K+\tau_a+K^2\eps/\kappa+\mathcal W\right)
 <\kappa/2
\]
for a sufficiently large constant $H_a$ depending only on $a$.
Choose $\tau_a,\zeta_a,e_a$ sufficiently small and $K_a$
sufficiently large to ensure this inequality and the estimates
above. Taking $e_a\le\min\{c_a,1/C_a\}$ and enlarging $K_a$
also ensures the hypotheses of Lemma~\ref{lem:pcp-waste}.
\end{proof}

For the fixed value $\tau_a$, the reference counts give
$N\ge\nu_aK^2n$ for some positive rational $\nu_a$ depending
only on $a$. Define
\[
 \eps_*(a)=\min\{1,\zeta_a/(2C_a),\nu_a e_a\}.
\]

\begin{proof}[Proof of Theorem~\ref{thm:soundness}]
Since $n\ge1$, the error $\eps=\eps_*(a)/N$ satisfies
$\eps\le e_a/K^2$. Lemma~\ref{lem:pcp-waste} gives
\[
 \mathcal W\le C_aN\eps=C_a\eps_*(a)\le\zeta_a/2.
\]
Lemma~\ref{lem:global-decoding} therefore proves correct decoding
at every vertex.
\end{proof}

\subsection{Completing the reduction}\label{sec:completion}

It remains to verify the encoding and totality requirements of
the search reduction. We first record a rounding lemma for general
HZ markets.

\begin{lemma}[Rational approximate witnesses]\label{lem:rational-witness}
\label{sec:rational-witnesses}
For rational $0<\eps\le1$, every $N$-buyer HZ market with utilities
in $[0,1]$ has a normalized additive $\eps$-equilibrium whose
coordinates have bit length $O(1+\log N+\log(1/\eps))$.
\end{lemma}
The proof is given in Appendix~\ref{app:rational-witness-proof}.

\begin{proof}[Proof of Theorem~\ref{thm:main}]
Membership in PPAD follows from~\cite[Section~7]{vy}, as explained
above. It remains to prove PPAD-hardness.
Fix $a=A/D$ and the constants of
Theorem~\ref{thm:robust-decoding}. Choose the least admissible
$K\ge K_a$ and use the market $\mathcal M(G;a,\tau_a,K)$.
All multiplicities are then constants depending only on $a$ and
the local degree pair.

Given a nonempty threshold instance, eliminate sinks as described
before~\eqref{eq:degrees}, then construct the market. Lemmas~\ref{lem:parameters}
and~\ref{lem:balancing} ensure positive integer sizes and equal
numbers of buyers and goods. Each vertex and edge contributes a
constant number of individual buyers and unit-supply goods. Thus
$N=O_a(|V|+|E|)$, and the explicit utility matrix is constructible
in polynomial time with entries in $\{0,a,1\}$.

The numbers of approvers of a good in the five module groups are
\[
 F_v+b_v+o_vK,\quad b_v+c_v+F_v,\quad c_v+h_vK,\quad
 A^R_{Fv}+F_v,\quad A^R_{Uv}+c_v.
\]
Their maximum over the six degree pairs depends only on $a$;
balancing goods have no approvers. Lemma~\ref{lem:rational-witness}
provides a normalized additive $\eps_*(a)/N$-equilibrium with
$O_a(1+\log N)$ bits per coordinate. Thus the search problem is
total and has witnesses of polynomial length.

For every returned equilibrium, Theorem~\ref{thm:soundness}
guarantees that the minimum-price decoder solves the augmented
threshold instance. Restricting to the original vertices solves
the input instance. All decoding operations are rational arithmetic
of polynomial complexity. For the empty graph, output the
one-buyer, one-good market with utility one and decode every
equilibrium to the empty assignment.
\end{proof}

\section{Constant-error hardness under PCP for PPAD}
\label{sec:pcp-equivalence}

The inverse-linear accuracy in Theorem~\ref{thm:main} can be
replaced by a positive constant under PCP for PPAD.

We use the polynomial-time formulation of the conjecture of
Babichenko, Papadimitriou, and Rubinstein~\cite{bpr}. It asserts
that, for some fixed positive error and failure tolerances,
finding a generalized-circuit assignment satisfying all but the
permitted fraction of gates is PPAD-hard~\cite{pcpmarkets}.
The original conjecture requires a quasilinear reduction.
OpenAI~\cite{openaipcp} has announced a proof of that formulation;
we retain the conjecture as a hypothesis.

The reduction passes through weak Pure-Circuit and weak
threshold games.

\subsection{Weak Pure-Circuit and threshold games}
\label{sec:pcp-forward}
\label{sec:pcp-sources}

We use NOR and PURIFY gates, as in~\cite{purecircuit}. This
formulation is equivalent to the NAND/PURIFY formulation
of~\cite{pcpmarkets}: replace NAND by NOR and complement every
Boolean value, leaving $\bot$ fixed. The transformation preserves
satisfaction of each gate and the number of times each node is
used as an input.

\begin{definition}[Pure-Circuit]\label{def:pcp-pure-circuit}
An instance consists of a finite node set $V_C$ and a nonempty
collection $\mathcal G$ of gates $(T,u,v,w)$, with pairwise distinct
$u,v,w\in V_C$ and $T\in\{\mathrm{NOR},\mathrm{PURIFY}\}$.
A NOR gate has inputs $u,v$ and output $w$; a PURIFY gate has input
$u$ and outputs $v,w$. Every node is produced by exactly one gate.

An assignment is a map $\xi:V_C\to\{0,1,\bot\}$. A NOR gate is
satisfied when
\[
 \begin{aligned}
 \xi_u=\xi_v=0&\quad\Longrightarrow\quad\xi_w=1,\\
 \xi_u=1\ \text{or}\ \xi_v=1
     &\quad\Longrightarrow\quad\xi_w=0.
 \end{aligned}
\]
Its output is unrestricted otherwise. A PURIFY gate is satisfied when
\[
 \begin{aligned}
 \xi_u=b\in\{0,1\}&\quad\Longrightarrow\quad\xi_v=\xi_w=b,\\
 \xi_u=\bot&\quad\Longrightarrow\quad
       \{\xi_v,\xi_w\}\cap\{0,1\}\ne\varnothing.
 \end{aligned}
\]
Pure-Circuit asks for an assignment satisfying every gate.
\end{definition}

\begin{definition}[Weak Pure-Circuit and the input-once restriction]
\label{def:weak-pure-circuit}
For $\delta\in[0,1]$, a weak $\delta$-Pure-Circuit solution fails
at most $\delta|\mathcal G|$ gates, counting each PURIFY gate once.
The case $\delta=0$ is ordinary Pure-Circuit. An instance is
\emph{input-once} if every node is used as an input to exactly one gate.
\end{definition}

By~\cite[Theorem 2.2]{pcpmarkets}, PCP for PPAD is equivalent
to PPAD-hardness of weak $\delta$-Pure-Circuit for some fixed
$\delta>0$, with or without the input-once restriction. We use
input-once instances so that the resulting threshold graph has
bounded degrees.

A threshold game, as in Definition~\ref{def:source}, is given by
a finite nonempty directed graph $G=(V,E)$, with
indegree $h_v\le2$ and outdegree $1\le o_v\le2$ at each vertex.
A threshold assignment is a vector $z\in[0,1]^V$. For a fixed
$\kappa>0$, the constraint at $v$ is
\[
 \begin{aligned}
 \sum_{u:u\to v}z_u>\tfrac12+\kappa
   &\quad\Longrightarrow\quad z_v\le\kappa,\\
 \sum_{u:u\to v}z_u<\tfrac12-\kappa
   &\quad\Longrightarrow\quad z_v\ge1-\kappa.
 \end{aligned}
\]
If the incoming sum lies in
$[1/2-\kappa,1/2+\kappa]$, any $z_v\in[0,1]$ is permitted.
An ordinary $\kappa$-threshold solution satisfies this constraint at
every vertex, as in Definition~\ref{def:source}.

\begin{definition}[Weak threshold solution]\label{def:weak-threshold}
For fixed $\kappa>0$ and $0\le\rho\le1$, a
$(\kappa,\rho)$-threshold solution is an assignment $z\in[0,1]^V$
whose threshold constraint fails at at most $\rho|V|$ vertices.
The weak threshold search problem asks for such an assignment given
$G$. For $\rho=0$, it is the ordinary threshold problem.
\end{definition}

Fix $\kappa=1/200$. The failure fractions $\delta$ and $\rho$
count circuit gates and threshold vertices, respectively.

\label{sec:pcp-threshold-construction}
\label{sec:pcp-constructions}
Given an input-once instance $\mathcal C=(V_C,\mathcal G)$,
construct the threshold graph using the gadgets
of~\cite[Theorem 4.23]{purecircuit}:
\begin{enumerate}
\item Create one graph vertex for each node of $V_C$, keeping its
      label. A circuit node shared by gates is represented by one vertex.
\item For a NOR gate with inputs $u,v$ and output $w$, add the edges
      $u\to w$ and $v\to w$.
\item For a PURIFY gate with input $u$ and outputs $v,w$, create four
      fresh vertices $s_1,s_2,s_3,s_4$, private to that gate. Add
      \[
      \begin{gathered}
       u\to s_1\to s_2\to s_3\to v,\qquad
       s_2\to s_4\to w,\\
       u\to s_3.
      \end{gathered}
      \]
\end{enumerate}
There are no other vertices or edges. The gadgets are shown in
Figure~\ref{fig:threshold-gadgets}.

\begin{figure}[htbp]
\centering
\begin{tikzpicture}[
  x=1cm,y=1cm,font=\normalsize,
  original/.style={circle,draw=blue!55!black,fill=blue!4,
    line width=.7pt,minimum size=.72cm,inner sep=2pt},
  fresh/.style={circle,draw=orange!65!black,fill=orange!5,dashed,
    line width=.7pt,minimum size=.72cm,inner sep=2pt},
  thresholdedge/.style={->,draw=black!70,line width=.75pt}
]
\node[font=\normalsize\bfseries] at (-5.2,1.85) {(a) NOR};
\node[original] (nor-u) at (-6.1,.7) {$u$};
\node[original] (nor-v) at (-6.1,-.7) {$v$};
\node[original] (nor-w) at (-4.3,0) {$w$};
\draw[thresholdedge] (nor-u)--(nor-w);
\draw[thresholdedge] (nor-v)--(nor-w);

\node[font=\normalsize\bfseries] at (1.7,1.85) {(b) PURIFY};
\node[original] (pur-u) at (-1.7,0) {$u$};
\node[fresh] (pur-s1) at (0,0) {$s_1$};
\node[fresh] (pur-s2) at (1.7,0) {$s_2$};
\node[fresh] (pur-s3) at (3.4,0) {$s_3$};
\node[original] (pur-v) at (5.1,0) {$v$};
\node[fresh] (pur-s4) at (3.4,-1.6) {$s_4$};
\node[original] (pur-w) at (5.1,-1.6) {$w$};
\draw[thresholdedge] (pur-u)--(pur-s1);
\draw[thresholdedge] (pur-s1)--(pur-s2);
\draw[thresholdedge] (pur-s2)--(pur-s3);
\draw[thresholdedge] (pur-s3)--(pur-v);
\draw[thresholdedge] (pur-s2)--(pur-s4);
\draw[thresholdedge] (pur-s4)--(pur-w);
\draw[thresholdedge] (pur-u.north)
  .. controls (-1.7,1.4) and (3.4,1.4) .. (pur-s3.north);
\end{tikzpicture}
\caption{The NOR and PURIFY threshold gadgets from the proof
of~\cite[Theorem 4.23]{purecircuit}, redrawn with our vertex labels.
Solid blue circles are original circuit nodes; dashed orange circles
are the four fresh vertices used for a PURIFY gate. Every arrow is a
directed edge of the threshold graph. For NOR, $u,v$ are the inputs
and $w$ is the output. For PURIFY, $u$ is the input and $v,w$ are
the outputs. The curved edge is $u\to s_3$. Only the edges contributed
by the displayed gate are shown; the remaining incidences of original
circuit nodes are omitted.}
\label{fig:threshold-gadgets}
\end{figure}
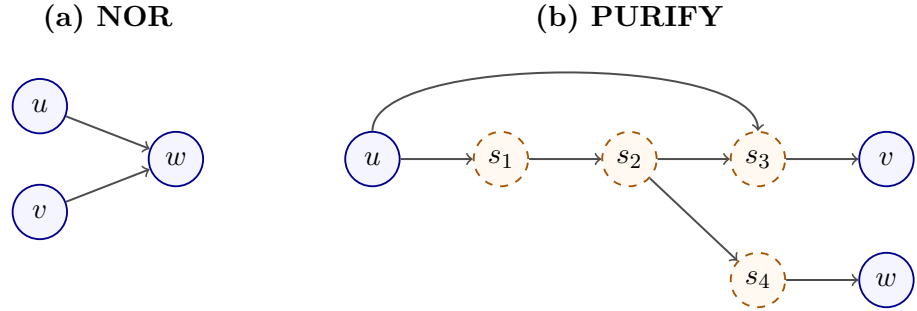

\FloatBarrier
\label{sec:pcp-threshold-decoding}
For a threshold assignment $z$, define
\begin{equation*}
 D_\kappa(s)=
 \begin{cases}
  0,&0\le s\le\kappa,\\
  1,&1-\kappa\le s\le1,\\
  \bot,&\kappa<s<1-\kappa.
 \end{cases}
\end{equation*}
and return $\xi_u=D_\kappa(z_u)$ for each original circuit node
$u\in V_C$.

\begin{lemma}[Weak threshold hardness]\label{lem:weak-threshold-source}
Assuming PCP for PPAD, there is a fixed rational $\rho>0$ for
which computing a $(\kappa,\rho)$-threshold solution is PPAD-hard,
even with $h_v\le2$ and $1\le o_v\le2$.
\end{lemma}
\begin{proof}
Under PCP for PPAD, fix $\delta>0$ for which weak
$\delta$-Pure-Circuit is hard on input-once
instances~\cite[Theorem 2.2]{pcpmarkets}. Let
$m=|\mathcal G|$. The local correctness argument
of~\cite[Theorem 4.23]{purecircuit} applies because
$\kappa<1/6$: a NOR gate is satisfied whenever its output
vertex satisfies its threshold constraint; a PURIFY gate is
satisfied whenever its four private vertices and two outputs
satisfy theirs. No threshold constraint at an input vertex is
needed for either conclusion.

Associate with each gate $g$ the set $T_g$ of its output
vertices and private vertices. These sets partition the graph:
each circuit node has one producing gate, and each additional
vertex belongs to one gadget. A NOR gate contributes one vertex
to its set and a PURIFY gate contributes six. Hence
\[
 |V|=\sum_{g\in\mathcal G}|T_g|\le6m.
\]
A failed decoded gate requires a failed threshold constraint
in its own set $T_g$. Since the sets are disjoint, choose a
positive rational $\rho\le\delta/6$ to obtain
\[
 \#\{\text{failed gates of }\mathcal C\}
 \le\#\{\text{failed threshold vertices}\}
 \le\rho|V|\le6\rho m\le\delta m.
\]
Thus every permitted threshold assignment decodes to a permitted
circuit assignment.

Every original node is an output of one gate and an input to
one gate. Its incoming edges belong to its producing gadget,
and its outgoing edges to its consuming gadget. Inspection of
the two gadgets gives $h_v\le2$ and $1\le o_v\le2$, including
at the private vertices. The graph and decoder have polynomial
size and are computable in polynomial time.
\end{proof}

\subsection{From weak threshold games to HZ}
\label{sec:pcp-threshold-to-hz}
\label{sec:pcp-market-construction}

Apply the construction of Section~\ref{sec:construction} to a
graph satisfying~\eqref{eq:degrees}, using the fixed reference
tolerance $\tau_a$ and a sufficiently large admissible $K$.
As before, decode by
\label{sec:pcp-decoding-goals}
\[
 p_v=\min_{j\in H_v}p_j,\qquad
 t_v=\frac{p_v-L}{w},\qquad z_v=\clip(t_v).
\]
At constant error, the bound $\mathcal W=O_a(N\eps)$ need
not make total unapproved expenditure small. We instead use the
price argument at vertices whose nearby groups receive little
unapproved expenditure, and bound the number of remaining vertices.
We first record the dependence of market size on $K$.

\begin{lemma}[Market size]\label{lem:uniform-market-size}
For sufficiently large admissible $K$ and $0<\tau\le\kappa/100$,
every nonempty $n$-vertex graph satisfying~\eqref{eq:degrees} gives
\begin{equation}\label{eq:pcp-market-size}
 N=\Theta_a(Kn/\eta)=\Theta_a(K^2n/\tau).
\end{equation}
For each fixed positive rational $\tau$ depending only on $a$,
there are positive rational $\nu_a,M_a$ such that
$\nu_aK^2n\le N\le M_aK^2n$ throughout this range.
\end{lemma}
\begin{proof}
At every vertex, $F_v,c_v=\Theta_a(K)$. Equation~\eqref{eq:reference-counts}
therefore gives reference supplies and anchor populations of size
$\Theta_a(K/\eta)$. All remaining populations contribute
$O_a(K)$ per vertex. Summing the counts in~\eqref{eq:N} and using
$\eta=2\tau/(\gamma K)$ proves~\eqref{eq:pcp-market-size}.
For fixed $\tau$, choose positive rational lower and upper
coefficients to obtain $\nu_a$ and $M_a$.
\end{proof}

Recall $\mathcal W_v$ from Subsection~\ref{sec:approximation}.
Let $B_2(v)$ be the vertices at undirected distance at most two
from $v$. The following local form of the decoding argument is
proved in Appendix~\ref{app:local-decoding-proof}.

\begin{lemma}[Local decoding]\label{lem:pcp-local}
There are positive constants $c_a,K_a$, depending only on $a$,
such that the following holds for admissible $K\ge K_a$,
$0<\tau\le\min\{\kappa/100,c_a\kappa\}$,
$\eps\le c_a\kappa^2/K^2$, and $\zeta\le c_a\kappa$.
For any normalized additive $\eps$-equilibrium, if
$\mathcal W_u\le\zeta$ for every $u\in B_2(v)$, then
\[
 \begin{aligned}
 \sum_{u\to v}z_u>\tfrac12+\kappa&\ \Longrightarrow\ z_v\le\kappa,\\
 \sum_{u\to v}z_u<\tfrac12-\kappa&\ \Longrightarrow\ z_v\ge1-\kappa.
 \end{aligned}
\]
Only the modules in $B_2(v)$ and their incident wire rows are
used; no bound is required on prices at exterior endpoints.
\end{lemma}

The proof uses the same price relations, expenditure identities,
and slack bounds as Lemma~\ref{lem:global-decoding}. The preliminary
price bounds at a vertex use only its own module; the price
relations also use its incoming neighbors. Applying those relations
at both endpoints of the wires incident to $v$ requires only
$B_2(v)$. Thus the argument needs small unapproved expenditure
on this neighborhood, even if prices elsewhere are unrestricted.

\begin{theorem}[Uniform robustness]\label{thm:robust-decoding}
For every fixed rational $a\in(0,1)$, there are positive rational
constants $\tau_a\le\kappa/100$, $e_a,B_a,\nu_a,M_a$, and an
integer $K_a\ge1$, depending only on $a$, such that the following
holds. For every admissible $K\ge K_a$ and every nonempty graph
satisfying~\eqref{eq:degrees}, the market
$\mathcal M(G;a,\tau_a,K)$, with $\eta=\tau_a w$, has
$\nu_aK^2n\le N\le M_aK^2n$. For every $0<\eps\le e_a/K^2$,
the minimum-price decoder applied to any normalized additive
$\eps$-equilibrium fails the $\kappa$-threshold condition at at most
\begin{equation*}
 B_aN\eps
\end{equation*}
vertices.
\end{theorem}

\begin{proof}
Choose the constants $\tau_a,\zeta_a,e_a$ in
Lemma~\ref{lem:global-decoding} sufficiently small, and $K_a$
sufficiently large, so that
\[
 \tau_a\le\min\{\kappa/100,c_a\kappa\},\qquad
 \zeta_a\le c_a\kappa,\qquad e_a\le c_a\kappa^2,
\]
where $c_a$ is from Lemma~\ref{lem:pcp-local}.
These choices preserve all conclusions of
Subsection~\ref{sec:approximation}. Lemma~\ref{lem:uniform-market-size}
gives the size bounds, using the same lower coefficient $\nu_a$
and a sufficiently large upper coefficient $M_a$.

Since $\sum_u\mathcal W_u\le\mathcal W$, at most
$\mathcal W/\zeta_a$ vertices have
$\mathcal W_u>\zeta_a$. The underlying undirected degree is
at most four, so at most $21$ vertices lie within distance two
of any one of them. Every vertex outside these neighborhoods
satisfies the hypothesis of Lemma~\ref{lem:pcp-local} and
decodes correctly. Lemma~\ref{lem:pcp-waste} therefore bounds
the number of failures by
\[
 \frac{21\mathcal W}{\zeta_a}
 \le\frac{21C_a}{\zeta_a}N\eps.
\]
Set $B_a=21C_a/\zeta_a$.
\end{proof}

\begin{proposition}[Forward reduction]\label{prop:pcp-forward}
There is a positive rational $s_a$, depending only on $a$, such that
for every rational $0<\rho\le1$ and every sufficiently large
admissible $K$, every normalized additive $s_a\rho/K^2$-equilibrium
of $\mathcal M(G;a,\tau_a,K)$ decodes to a
$(\kappa,\rho)$-threshold solution.
\end{proposition}
\begin{proof}
Let $e_a,B_a,M_a$ be the constants in
Theorem~\ref{thm:robust-decoding} and
Lemma~\ref{lem:uniform-market-size}. Choose a positive rational
\[
 s_a\le\min\{e_a,1/(B_aM_a)\}.
\]
For $\eps=s_a\rho/K^2$, uniform robustness applies because
$\eps\le e_a/K^2$. Together with the market-size bound, it gives
\[
 \frac{\#\{\text{failed threshold vertices}\}}{n}
 \le B_a\frac{N}{n}\eps
 \le B_aM_aK^2\eps
 =B_aM_as_a\rho\le\rho.
\]
For fixed parameters the market has $N=O_{a,K}(n)$ buyers
and goods, so its explicit utility matrix has polynomial size.
Construction and decoding use rational arithmetic of polynomial
bit complexity.
\end{proof}

\begin{proof}[Proof of Theorem~\ref{thm:pcp-intro}]
Assume PCP for PPAD and choose a fixed rational $0<\rho\le1$
for which the weak threshold problem is PPAD-hard, as supplied
by Lemma~\ref{lem:weak-threshold-source}. Fix a sufficiently
large admissible $K$ and put $\eps_0=s_a\rho/K^2>0$.
Proposition~\ref{prop:pcp-forward} converts every
$\eps_0$-equilibrium of the constructed market into a
$(\kappa,\rho)$-threshold solution. The construction has
polynomial size, and Lemma~\ref{lem:rational-witness} gives
totality and polynomial-length rational outputs. Every utility
belongs to $\{0,a,1\}$, proving the theorem. A PTAS run at
this fixed accuracy would solve the problem in polynomial time. Hence, assuming PCP for PPAD, this problem admits
no PTAS unless $\mathrm{PPAD}\subseteq\mathrm{P}$.
\end{proof}

\section{Two common marginal slopes}\label{sec:two-slopes}

Let $\mathcal B$ be the set of buyers and $\mathcal G$ the set of
goods, with $|\mathcal B|=|\mathcal G|=N$. Budgets, supplies, and
required bundle sizes are one, as in Section~\ref{sec:prelim}.
Throughout this section, buyer indices $i$ range over $\mathcal B$
and good indices $j$ range over $\mathcal G$, unless a smaller set
is specified.

We prove Theorem~\ref{thm:two-slopes} by giving a strongly polynomial-time
algorithm for an exact normalized equilibrium with
minimum-expenditure tie-breaking. We first state the search problem
and reduce it to threshold utilities. We then describe the flow
algorithm, prove its correctness, and bound its arithmetic operation
count and intermediate encoding lengths.

\subsection{Problem definition and normalization}

\begin{definition}[Two-slope HZ equilibrium]
\label{def:two-slopes-problem}
For each buyer $i\in\mathcal B$, the input specifies rational slopes
$c_i^{(1)}\ge c_i^{(2)}\ge0$ and rational breakpoints
$\tau_{ij}\in[0,1]$ for every good $j\in\mathcal G$. The slopes are common
across goods for a fixed buyer but may differ between buyers. Utilities
are separable:
\[
 U_i(x_i)=\sum_{j\in\mathcal G} f_{ij}(x_{ij}),\qquad
 f_{ij}(t)=c_i^{(2)}t+
       \bigl(c_i^{(1)}-c_i^{(2)}\bigr)\min\{t,\tau_{ij}\},
 \quad 0\le t\le1.
\]

\begin{figure}[htbp]
\centering
\begin{tikzpicture}[
  x=4.5cm,y=4.5cm,>=stealth,
  utilityaxis/.style={->,draw=black!65,line width=.65pt},
  utilityguide/.style={draw=black!35,dashed,line width=.5pt},
  firstslope/.style={draw=blue!55!black,line width=1.2pt},
  secondslope/.style={draw=orange!65!black,line width=1.2pt}
]
\begin{scope}
  \coordinate (origin) at (0,0);
  \coordinate (knee) at (.45,.63);
  \coordinate (endpoint) at (1,.8225);
  \node at (.5,1.28) {(a) Original utility};
  \draw[utilityaxis] (-.04,0) -- (1.08,0) node[right] {$t$};
  \draw[utilityaxis] (0,-.04) -- (0,1.08)
    node[above] {$f_{ij}(t)$};
  \node[below left] at (origin) {$0$};
  \draw[draw=black!65] (1,.02) -- (1,-.02) node[below] {$1$};
  \draw[utilityguide] (knee) -- (.45,0);
  \draw[draw=black!65] (.45,.02) -- (.45,-.02)
    node[below] {$\tau_{ij}$};
  \draw[firstslope] (origin) --
    node[midway,above,sloped,text=blue!55!black] {$c_i^{(1)}$}
    (knee);
  \draw[secondslope] (knee) --
    node[midway,above,sloped,text=orange!65!black] {$c_i^{(2)}$}
    (endpoint);
  \fill[black!70] (knee) circle[radius=1.5pt];
\end{scope}
\begin{scope}[xshift=7cm]
  \coordinate (origin) at (0,0);
  \coordinate (knee) at (.45,.45);
  \coordinate (endpoint) at (1,.45);
  \node at (.5,1.28) {(b) Normalized utility};
  \draw[utilityaxis] (-.04,0) -- (1.08,0) node[right] {$t$};
  \draw[utilityaxis] (0,-.04) -- (0,1.08)
    node[above] {$\widetilde f_{ij}(t)$};
  \node[below left] at (origin) {$0$};
  \draw[draw=black!65] (1,.02) -- (1,-.02) node[below] {$1$};
  \draw[utilityguide] (knee) -- (.45,0);
  \draw[utilityguide] (knee) -- (0,.45)
    node[left] {$\tau_{ij}$};
  \draw[draw=black!65] (.45,.02) -- (.45,-.02)
    node[below] {$\tau_{ij}$};
  \draw[firstslope] (origin) --
    node[midway,above,sloped,text=blue!55!black] {$1$}
    (knee);
  \draw[secondslope] (knee) --
    node[midway,above,text=orange!65!black] {$0$}
    (endpoint);
  \fill[black!70] (knee) circle[radius=1.5pt];
\end{scope}
\end{tikzpicture}
\caption{A two-slope utility and its normalized threshold utility
for the same buyer $i$, good $j$, and breakpoint $\tau_{ij}$.
The left panel illustrates $c_i^{(1)}>c_i^{(2)}>0$.
For $c_i^{(1)}>c_i^{(2)}$, subtracting $c_i^{(2)}t$ and
dividing by $c_i^{(1)}-c_i^{(2)}$ gives
$\widetilde f_{ij}(t)=\min\{t,\tau_{ij}\}$, shown on the right
with slope one (a $45$-degree segment) followed by slope zero.
On unit bundles this induces a positive affine transformation of
the buyer's total utility.}
\label{fig:two-slope-utilities}
\end{figure}
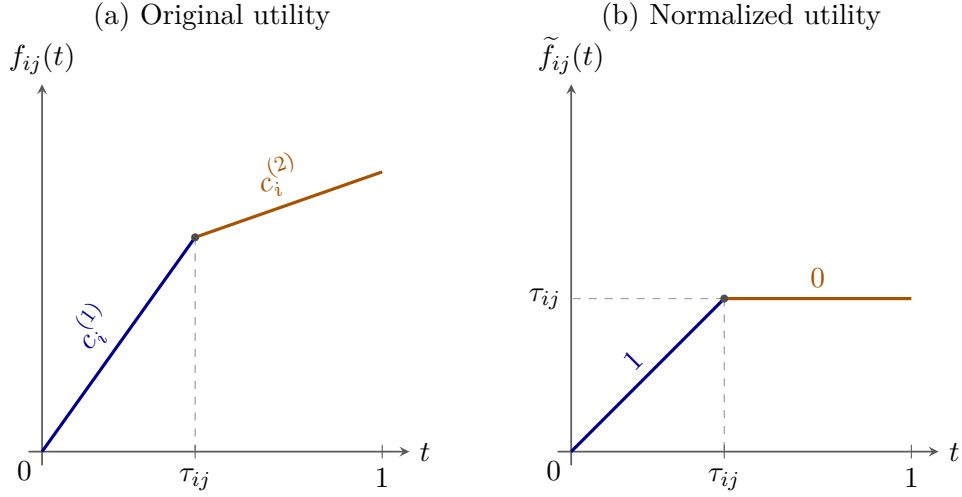

The problem is to compute an exact normalized HZ equilibrium with
minimum-expenditure tie-breaking, as defined in Section~\ref{sec:prelim}.
Here the demand value $V_i(p)$ is given by~\eqref{eq:value} with
objective $U_i(y)$ in place of $\sum_j u_{ij}y_j$; its budget,
unit-quantity, and nonnegativity constraints are unchanged.
\end{definition}

On a unit bundle, the utility in
Definition~\ref{def:two-slopes-problem} satisfies
\[
 U_i(x_i)=c_i^{(2)}+
       \bigl(c_i^{(1)}-c_i^{(2)}\bigr)
       \sum_j\min\{x_{ij},\tau_{ij}\}.
\]
If $c_i^{(1)}>c_i^{(2)}$, a positive affine normalization therefore
preserves both utility-maximizing bundles and minimum-expenditure
tie-breaking. If the slopes are equal, the buyer is indifferent
among unit bundles; we represent her by replacing her entire
threshold row with zeros. Hence it suffices to solve the problem for
the normalized \emph{threshold utilities}
\begin{equation}\label{eq:threshold-utility}
 U_i(x_i)=\sum_j\min\{x_{ij},\tau_{ij}\}.
\end{equation}
Throughout the two-slope analysis, $U_i$ and $\tau_{ij}$ denote these normalized utilities
and thresholds.

Each of the first $\tau_{ij}$ units of good $j$ contributes one unit
of utility per unit of quantity to buyer $i$; additional units
contribute no utility. We call the former quantities \emph{useful}.
The algorithm assigns useful quantities in increasing order of
price, then fills any remaining bundle deficits with free goods.
For thresholds in $\{0,1\}$, receiving a positive threshold fills
the buyer's bundle; this is the bivalued case treated
by~\cite[Section 4]{vy}. With fractional thresholds, a buyer may
continue through several price levels, so the algorithm must track
both her remaining budget and her remaining quantity.

\subsection{The flow algorithm}

\paragraph{State and invariant.}
Let $\mathcal J\subseteq\mathcal G$ be the goods whose prices have not yet been
assigned. These goods are still unallocated; trial flows will not
change the current allocation $x$. Define each buyer's remaining
budget and quantity by
\[
 B_i=1-\sum_{j\in\mathcal G\setminus\mathcal J}p_jx_{ij},\qquad
 R_i=1-\sum_{j\in\mathcal G}x_{ij},
\]
and let
\[
 \mathcal I=\{i\in\mathcal B:B_i>0\text{ and }R_i>0\}
\]
be the active buyers. Initially, $x=0$, $B_i=R_i=1$, and
$\mathcal I=\mathcal B$, $\mathcal J=\mathcal G$. We update $\mathcal I$ after
committing each price group, using the updated budgets and quantities.

Before the final completion step, we maintain $x_{ij}\le\tau_{ij}$
and the invariant
\begin{equation}\label{eq:threshold-invariant}
 x_{ij}=\tau_{ij}
 \qquad\text{for every }i\in\mathcal I
 \text{ and every }j\in\mathcal G\setminus\mathcal J.
\end{equation}
Thus an active buyer has already received every useful unit available
to her at earlier prices. Every good assigned a positive price is
allocated in full. Any unused supply of free goods is reserved for
completion.

\paragraph{The expenditure network.}
For $\lambda>0$, form a network $\mathcal F(\lambda)$ with
source $s$, sink $t$, goods $j\in\mathcal J$, and buyers
$i\in\mathcal I$. Its capacities are
\[
 \begin{array}{ccl}
 s\to j&:&\lambda,\\
 j\to i&:&\lambda\tau_{ij},\\
 i\to t&:&\min\{B_i,\lambda R_i\}.
 \end{array}
\]

Every ordinary maximum-flow computation below starts from zero and
uses shortest augmenting paths, found by breadth-first search.

Write $w_{ij}$ for the flow on $j\to i$, keeping the buyer-first
indexing used in $x_{ij}$ and $\tau_{ij}$. It represents expenditure,
so the corresponding quantity is $w_{ij}/\lambda$. The three arc
types enforce supply, useful-quantity limits, and the buyer's budget
and bundle-size limits, respectively. We set $w_{ij}=0$ for buyers
outside $\mathcal I$ when referring to a full flow matrix.

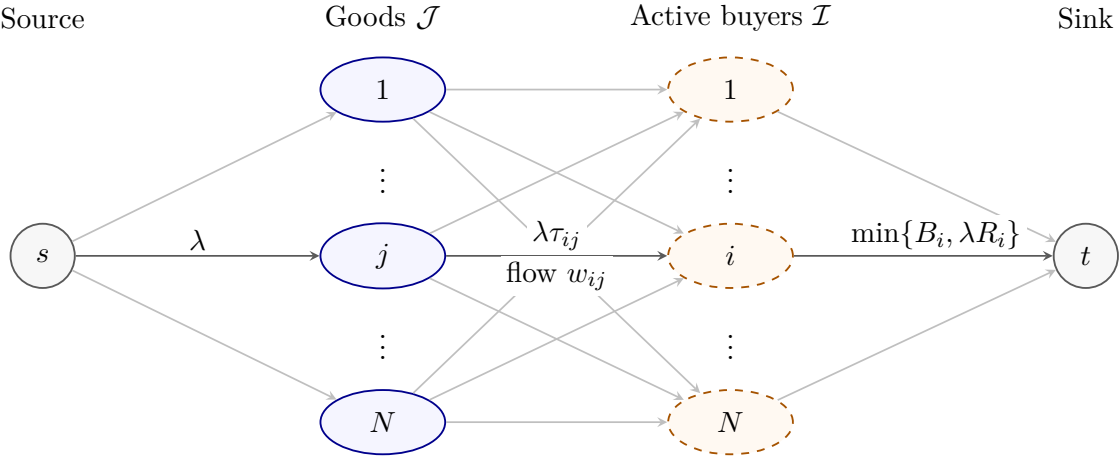
\begin{figure}[htbp]
\centering
\begin{tikzpicture}[
  x=1cm,y=1cm,>=stealth,
  hzflow/.style={hzedge,->},
  hzotherflow/.style={hzflow,draw=black!25},
  hzterminal/.style={circle,draw=black!65,fill=black!3,
    line width=.7pt,minimum size=.85cm,inner sep=2pt}
]
\node at (-6.9,3.15) {Source};
\node at (-2.4,3.15) {Goods $\mathcal J$};
\node at (2.2,3.15) {Active buyers $\mathcal I$};
\node at (6.9,3.15) {Sink};

\node[hzterminal] (source) at (-6.9,0) {$s$};
\node[hzterminal] (sink) at (6.9,0) {$t$};

\node[hzgood] (firstgood) at (-2.4,2.2) {$1$};
\node[hzgood] (good) at (-2.4,0) {$j$};
\node[hzgood] (lastgood) at (-2.4,-2.2) {$N$};

\node[hzbuyer] (firstbuyer) at (2.2,2.2) {$1$};
\node[hzbuyer] (buyer) at (2.2,0) {$i$};
\node[hzbuyer] (lastbuyer) at (2.2,-2.2) {$N$};

\foreach \y in {1.1,-1.1} {
  \node at (-2.4,\y) {$\vdots$};
  \node at (2.2,\y) {$\vdots$};
}

\foreach \g in {firstgood,lastgood}
  \draw[hzotherflow] (source) -- (\g);

\foreach \b in {firstbuyer,lastbuyer}
  \draw[hzotherflow] (\b) -- (sink);

\foreach \g/\b in {
  firstgood/firstbuyer,firstgood/buyer,firstgood/lastbuyer,
  good/firstbuyer,good/lastbuyer,
  lastgood/firstbuyer,lastgood/buyer,lastgood/lastbuyer%
} {
  \draw[hzotherflow] (\g) -- (\b);
}

\draw[hzflow] (source) --
  node[hzcoef,above] {$\lambda$} (good);

\draw[hzflow] (good) --
  node[hzcoef,above] {$\lambda\tau_{ij}$}
  node[hzcoef,below] {flow $w_{ij}$} (buyer);

\draw[hzflow] (buyer) --
  node[hzcoef,above,pos=.55] {$\min\{B_i,\lambda R_i\}$} (sink);

\end{tikzpicture}
\caption{The expenditure network $\mathcal F(\lambda)$.
Solid blue ellipses denote unpriced goods, and dashed orange
ellipses denote active buyers. Labels above the highlighted arcs
give their capacities; lighter arcs have the corresponding
capacities for their endpoints. A flow $w_{ij}$ on $j\to i$
represents expenditure and allocates $w_{ij}/\lambda$ units of
good $j$ to buyer $i$. The original indices $1,\ldots,N$ are shown
schematically; each stage retains only goods in $\mathcal J$
and buyers in $\mathcal I$.}
\label{fig:threshold-network}
\end{figure}

In the residual network, unused capacity gives a forward arc and
positive flow gives a reverse arc. A good is \emph{closed} if it has
no residual path to $t$. Closed goods may still have individual
residual arcs. To find them, reverse all residual arcs, search from
$t$, and select the goods not reached by the search. The pseudocode
denotes this operation by $\textsc{ClosedGoods}$.

\paragraph{Selecting a positive price.}
All remaining goods can be allocated usefully at a common price
$\lambda$ if and only if the maximum flow has value
$\lambda|\mathcal J|$. For reference, the largest such price is
characterized by
\begin{equation}\label{eq:threshold-price-lp}
\begin{aligned}
 \max_{\lambda,w}\ &\lambda\\
 \text{subject to}\quad
 &\sum_{i\in\mathcal I}w_{ij}=\lambda
       &&\text{for }j\in\mathcal J,\\
 &\sum_{j\in\mathcal J}w_{ij}\le B_i
       &&\text{for }i\in\mathcal I,\\
 &\sum_{j\in\mathcal J}w_{ij}\le\lambda R_i
       &&\text{for }i\in\mathcal I,\\
 &0\le w_{ij}\le\lambda\tau_{ij}
       &&\text{for }i\in\mathcal I,\ j\in\mathcal J,\\
 &\lambda\ge0.
\end{aligned}
\end{equation}
We compute this optimum without invoking a linear-programming
algorithm, using the parametric maximum-flow algorithm of
Gallo, Grigoriadis, and Tarjan~\cite{ggt}.
We commit only the closed goods of the returned expenditure flow.
Buyers who remain
active receive their full thresholds on those goods, preserving
\eqref{eq:threshold-invariant}; the other trial allocations are
discarded. The correctness proof shows that the closed set is nonempty
and that the next price is strictly larger.

\paragraph{The price-finding network.}
For $\vartheta\ge0$, form a quantity network $\mathcal Q(\vartheta)$
with source $s$, sink $t$, two vertices $i^-,i^+$ for each
$i\in\mathcal I$, and one vertex for each $j\in\mathcal J$.
Its arcs and capacities are
\begin{equation}\label{eq:threshold-price-network}
 \begin{array}{ccl}
 s\to i^-&:&B_i\vartheta,\\
 i^-\to i^+&:&R_i,\\
 i^+\to j&:&\tau_{ij},\\
 j\to t&:&1.
 \end{array}
\end{equation}
Only the source capacities depend on $\vartheta$; the active sets,
budgets, and remaining quantities stay fixed during a price search.
The two consecutive buyer arcs enforce both affordability and the
remaining quantity limit. Figure~\ref{fig:threshold-price-network}
illustrates this auxiliary network.

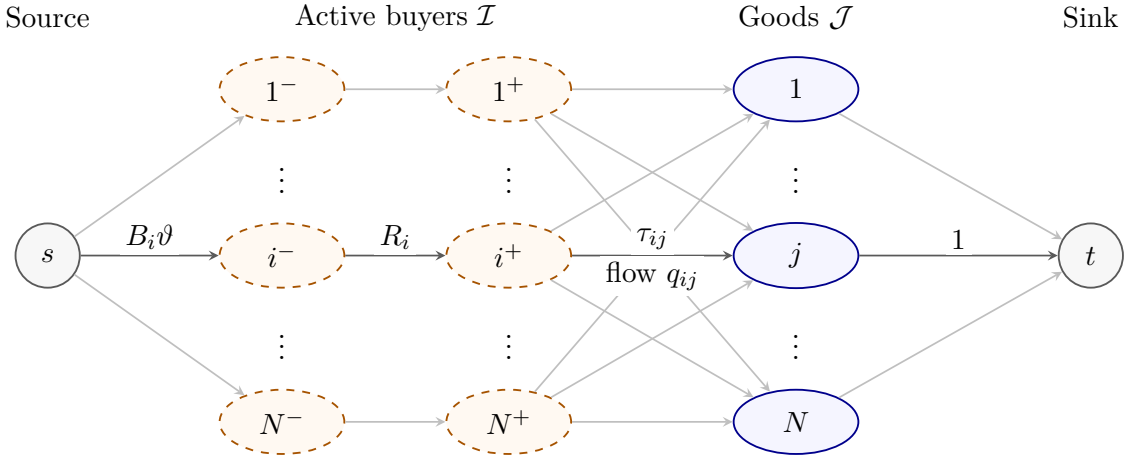
\begin{figure}[htbp]
\centering
\begin{tikzpicture}[
  x=1cm,y=1cm,>=stealth,
  hzflow/.style={hzedge,->},
  hzotherflow/.style={hzflow,draw=black!25},
  hzterminal/.style={circle,draw=black!65,fill=black!3,
    line width=.7pt,minimum size=.85cm,inner sep=2pt}
]
\node at (-6.9,3.15) {Source};
\node at (-2.3,3.15) {Active buyers $\mathcal I$};
\node at (3.0,3.15) {Goods $\mathcal J$};
\node at (6.9,3.15) {Sink};

\node[hzterminal] (source) at (-6.9,0) {$s$};
\node[hzterminal] (sink) at (6.9,0) {$t$};

\node[hzbuyer] (firstminus) at (-3.8,2.2) {$1^-$};
\node[hzbuyer] (minus) at (-3.8,0) {$i^-$};
\node[hzbuyer] (lastminus) at (-3.8,-2.2) {$N^-$};

\node[hzbuyer] (firstplus) at (-0.8,2.2) {$1^+$};
\node[hzbuyer] (plus) at (-0.8,0) {$i^+$};
\node[hzbuyer] (lastplus) at (-0.8,-2.2) {$N^+$};

\node[hzgood] (firstgood) at (3.0,2.2) {$1$};
\node[hzgood] (good) at (3.0,0) {$j$};
\node[hzgood] (lastgood) at (3.0,-2.2) {$N$};

\foreach \y in {1.1,-1.1} {
  \node at (-3.8,\y) {$\vdots$};
  \node at (-0.8,\y) {$\vdots$};
  \node at (3.0,\y) {$\vdots$};
}

\foreach \bm/\bp in {firstminus/firstplus,lastminus/lastplus} {
  \draw[hzotherflow] (source) -- (\bm);
  \draw[hzotherflow] (\bm) -- (\bp);
}

\foreach \bp/\g in {
  firstplus/firstgood,firstplus/good,firstplus/lastgood,
  plus/firstgood,plus/lastgood,
  lastplus/firstgood,lastplus/good,lastplus/lastgood%
} {
  \draw[hzotherflow] (\bp) -- (\g);
}

\foreach \g in {firstgood,lastgood}
  \draw[hzotherflow] (\g) -- (sink);

\draw[hzflow] (source) --
  node[hzcoef,above] {$B_i\vartheta$} (minus);

\draw[hzflow] (minus) --
  node[hzcoef,above] {$R_i$} (plus);

\draw[hzflow] (plus) --
  node[hzcoef,above] {$\tau_{ij}$}
  node[hzcoef,below] {flow $q_{ij}$} (good);

\draw[hzflow] (good) --
  node[hzcoef,above] {$1$} (sink);

\end{tikzpicture}
\caption{The quantity network $\mathcal Q(\vartheta)$ for price
selection, with $\vartheta=1/\lambda$. Each active buyer is split
into two orange vertices; blue vertices represent unpriced goods.
Highlighted arcs are labeled with their capacities, and only
$B_i\vartheta$ varies during the search. The smallest $\vartheta$
permitting flow $|\mathcal J|$ gives the largest feasible price
$\lambda=1/\vartheta$.}
\label{fig:threshold-price-network}
\end{figure}

\begin{lemma}[Strongly polynomial price selection]
\label{lem:threshold-strong-price}
Suppose $\mathcal J\ne\varnothing$ and a positive price is feasible
in~\eqref{eq:threshold-price-lp}. Its optimum and a corresponding
expenditure flow can be computed by one call to the GGT parametric
maximum-flow algorithm and one ordinary maximum-flow computation,
using a number of arithmetic operations polynomial in $N$ and
independent of the numerical data.
\end{lemma}
\begin{proof}
Fix the active sets, budgets, and remaining quantities for this
stage, and put $k=|\mathcal J|$. For $\vartheta>0$, a flow of
value $k$ in $\mathcal Q(\vartheta)$ saturates every arc $j\to t$.
Writing $q_{ij}$ for its flow on $i^+\to j$, we have
\[
 \sum_i q_{ij}=1,\qquad
 \sum_j q_{ij}\le\min\{B_i\vartheta,R_i\},\qquad
 0\le q_{ij}\le\tau_{ij}.
\]
Thus $\lambda=1/\vartheta$ and $w_{ij}=\lambda q_{ij}$ give
a feasible expenditure flow at price $\lambda$. Conversely,
dividing a feasible expenditure flow by $\lambda$ gives a
value-$k$ flow in $\mathcal Q(1/\lambda)$.

Let $F(\vartheta)$ be the maximum-flow value of
$\mathcal Q(\vartheta)$, and put
$T=\max_{i\in\mathcal I}R_i/B_i$.
At $\vartheta=T$, each source capacity is at least the following
quantity capacity, so the budget arcs impose no additional
restriction. Positive-price feasibility therefore implies
$F(T)=k$, whereas $F(0)=0$.
By max-flow/min-cut, $F$ is the minimum of finitely many affine
cut-capacity functions, and hence is continuous and piecewise
linear. It is nondecreasing and bounded above by $k$.
Consequently,
\[
 \vartheta_*=
 \min\{\vartheta\in[0,T]:F(\vartheta)=k\},\qquad
 \lambda_* = 1/\vartheta_*
\]
are well-defined, and the total source capacity gives
$\vartheta_*\ge k/\sum_iB_i>0$.

The capacities leaving the source are nondecreasing affine
functions of $\vartheta$, while all other capacities are constant.
These are the hypotheses of the GGT breakpoint algorithm
(Theorem~3.1 of their paper). It computes the breakpoints and
linear pieces of $F$ on $[0,T]$ in strongly polynomial time,
and hence the first parameter $\vartheta_*$ for which $F=k$.
The network has $O(N)$ vertices and $O(N^2)$ arcs.
Finally, compute an ordinary maximum flow $q$ from zero in
$\mathcal Q(\vartheta_*)$ and return
$\lambda=1/\vartheta_*$ and $w_{ij}=\lambda q_{ij}$.
This returned $w$ is a maximum flow of value $\lambda k$ in
the original expenditure network $\mathcal F(\lambda)$.
\end{proof}

We denote this procedure by
$\textsc{Price}(\mathcal I,\mathcal J,B,R)$.
The goods to commit are determined by residual reachability in
$\mathcal F(\lambda)$ using the returned expenditure flow,
not by cuts of $\mathcal Q(\vartheta)$.

\begin{center}
\begingroup
\setlength{\fboxsep}{10pt}
\fbox{%
\begin{minipage}{\dimexpr\linewidth-2\fboxsep-2\fboxrule\relax}
\refstepcounter{algorithm}\label{alg:threshold-equilibrium}
\textbf{Algorithm \thealgorithm.} Strongly polynomial equilibrium for threshold utilities
\medskip

\noindent\textbf{Input:} Rational thresholds $\tau_{ij}\in[0,1]$
for $i\in\mathcal B$, $j\in\mathcal G$.

\smallskip
\noindent\textbf{Output:} A normalized equilibrium $(x,p)$ with
minimum-expenditure tie-breaking.

\begin{enumerate}[label=\arabic*.,leftmargin=1.8em,
                  topsep=6pt,itemsep=5pt,parsep=0pt]
\item Initialize $x\gets0$, $\mathcal I\gets\mathcal B$, and
      $\mathcal J\gets\mathcal G$.
      Set $B_i\gets1$ and $R_i\gets1$ for every buyer $i\in\mathcal B$.

\item \emph{Phase 1: allocate free goods.}
  \begin{enumerate}[label=(\alph*),leftmargin=2em,
                    topsep=3pt,itemsep=2pt,parsep=0pt]
  \item Compute a maximum flow $w$ in the initial quantity network
        $\mathcal F(1)$.
  \item Let $S_0$ be its closed goods.
  \item Assign $p_j\gets0$ and $x_{ij}\gets w_{ij}$ for
        $j\in S_0$, $i\in\mathcal B$.
  \item Update $R_i\gets R_i-\sum_{j\in S_0}x_{ij}$ for each buyer;
        budgets are unchanged.
  \item Remove $S_0$ from $\mathcal J$ and update the active buyers
        $\mathcal I$.
  \end{enumerate}

\item \emph{Phase 2: allocate goods at increasing positive prices.}
  \begin{enumerate}[label=(\alph*),leftmargin=2em,
                    topsep=3pt,itemsep=2pt,parsep=0pt]
  \item While $\mathcal J\ne\varnothing$, do:
    \begin{enumerate}[label=\roman*.,leftmargin=1.8em,
                      topsep=3pt,itemsep=2pt,parsep=0pt]
    \item Compute $(\lambda,w)\gets
          \textsc{Price}(\mathcal I,\mathcal J,B,R)$ using
          Lemma~\ref{lem:threshold-strong-price}.
    \item Let $S$ be the closed goods of this flow in
          $\mathcal F(\lambda)$.
    \item Assign $p_j\gets\lambda$ and $x_{ij}\gets w_{ij}/\lambda$
          for $j\in S$, $i\in\mathcal I$.
    \item Let $d_i=\sum_{j\in S}x_{ij}$ for each active buyer $i$.
    \item Update $(B_i,R_i)\gets(B_i-\lambda d_i,\ R_i-d_i)$
          for each active buyer.
    \item Remove $S$ from $\mathcal J$ and update the active buyers
          $\mathcal I$.
    \item Discard the trial allocations on goods outside $S$.
    \end{enumerate}
  \end{enumerate}

\item \emph{Phase 3: fill the remaining bundles.}
      Use the unused supply of $S_0$ to fill all remaining bundle
      deficits, including those of buyers whose budgets are exhausted.

\item Return $(x,p)$.
\end{enumerate}
\end{minipage}%
}
\endgroup
\end{center}

The free stage leaves all budgets equal to one. In the positive-price
loop, each update uses the active set from the start of that stage;
we recompute the active set only after committing the group. In the
completion phase, we include all buyers with a quantity deficit,
even if their budgets are exhausted. Lemma~\ref{lem:threshold-completion}
shows that free supply is sufficient and that these final quantities
change neither utility nor expenditure.

\subsection{Correctness}

We first prove feasibility and progress of the price stages. We then
show that completion preserves utility and expenditure, and establish
optimal demand with minimum-expenditure tie-breaking.

\subsubsection{Feasibility and progress}

We first record the residual-network property that determines which
allocations can be fixed at a given price.

\begin{lemma}[Closed goods]\label{lem:threshold-closed}
Let $S$ be the closed goods of a maximum flow in
$\mathcal F(\lambda)$. Every good outside $S$ has a saturated
arc from $s$. For each active buyer $i\in\mathcal I$, at least one of the following holds:
\begin{enumerate}
\item every arc $j\to i$ with $j\in S$ is saturated;
\item the arc $i\to t$ is saturated, and all positive flow into
      $i$ comes from $S$.
\end{enumerate}
\end{lemma}
\begin{proof}
An unsaturated arc $s\to j$ with $j\notin S$ could be followed
by a residual path from $j$ to $t$, contradicting maximality.
Now suppose some arc $j\to i$ with $j\in S$ is unsaturated.
The buyer $i$ cannot reach $t$ in the residual network, since then
$j$ could also reach $t$. In particular, $i\to t$ is saturated.
If $i$ received positive flow from $k\notin S$, the reverse arc
$i\to k$ followed by a residual path from $k$ to $t$ would give
the same contradiction. Thus all her inflow comes from $S$.
\end{proof}

\begin{lemma}[Largest feasible price]\label{lem:threshold-price}
Suppose $\mathcal J\ne\varnothing$ and
\eqref{eq:threshold-price-lp} is feasible at some positive price.
It has a finite positive optimum. Every optimal flow saturates all
source arcs and has a nonempty set of closed goods. Moreover, if a
feasible positive price admits a flow with no closed goods, then a
strictly larger price is feasible.
\end{lemma}
\begin{proof}
The budget constraints give
$\lambda|\mathcal J|\le\sum_{i\in\mathcal I}B_i\le N$. Both $\lambda$ and
all $w_{ij}$ are therefore bounded. The feasible region is closed
and nonempty, so an optimum exists and is positive. The good
equalities saturate every source arc, making the corresponding flow
a maximum flow of value $\lambda|\mathcal J|$.

For the final assertion, consider such a flow at price
$\lambda_0>0$. If every good has a residual path to $t$, then
every cut with a good on its source side has positive residual
capacity and hence capacity strictly greater than
$\lambda_0|\mathcal J|$. There are finitely many cuts, and their
capacities depend continuously on $\lambda$. These inequalities
remain strict after a sufficiently small increase in $\lambda$.
A cut with no goods on its source side contains all source arcs, so
its capacity is at least $\lambda|\mathcal J|$ at every price.
Max-flow/min-cut therefore gives a larger feasible price. At an
optimum this is impossible, proving that a closed good exists.
\end{proof}

\begin{lemma}[The free stage]\label{lem:threshold-free}
The set $S_0$ is nonempty. After the free stage, every buyer with
$R_i>0$ has received $\tau_{ij}$ units of each $j\in S_0$.
If any goods remain, the first positive-price stage is feasible at
one and has optimum strictly greater than one.
\end{lemma}
\begin{proof}
The sink capacity of every buyer in $\mathcal F(1)$ is one.
Lemma~\ref{lem:threshold-closed} implies that each buyer receives
either all her thresholds on $S_0$ or one useful unit entirely from
$S_0$. In the second case $R_i=0$, so the invariant holds.

If the maximum-flow value is less than $N$, some source arc is
unsaturated, and its good belongs to $S_0$ by
Lemma~\ref{lem:threshold-closed}. If the value is $N$, all sink
arcs are saturated and no good can reach $t$. Thus $S_0$ is
nonempty in either case.

Restrict the flow to the remaining goods. Their source arcs are
saturated. A buyer receiving $d_i$ free units retains inflow at
most $1-d_i=R_i$ and budget one, so the restricted flow is feasible
at price one. Buyers with $R_i=0$ have no remaining inflow and can
be deleted.

Every remaining good also retains a residual path to $t$. Its old
path avoids closed goods and $s$, since neither could reach $t$.
It avoids deleted buyers as well: their entire inflow came from
$S_0$ and their sink arcs were saturated. For a retained buyer,
both sink capacity and sink flow decrease by $d_i$, preserving
forward residual capacity. All other arcs on the path are unchanged.
Lemma~\ref{lem:threshold-price} now gives a feasible price greater
than one.
\end{proof}

\begin{lemma}[Committing a positive-price group]
\label{lem:threshold-stage}
At each positive-price stage, $S$ is nonempty and every good in
$S$ is allocated in full. Every buyer who remains active receives
her full threshold from each good in $S$. If goods remain, the
next positive-price stage is feasible and has a strictly larger optimum.
There are at most $N$ positive-price stages.
\end{lemma}
\begin{proof}
The first stage is feasible by Lemma~\ref{lem:threshold-free}.
At any feasible stage, Lemma~\ref{lem:threshold-price} gives a
nonempty closed set and saturated source arcs. Dividing their flows
by $\lambda$ allocates each good in $S$ in full.

If a buyer does not receive all her thresholds on $S$,
Lemma~\ref{lem:threshold-closed} gives
\[
 \lambda d_i=\min\{B_i,\lambda R_i\}.
\]
Her budget or quantity is exhausted. She leaves the active set,
and every buyer who remains satisfies the invariant.

Restrict the flow to the remaining goods. At price $\lambda$,
each buyer's sink flow decreases by $\lambda d_i$, exactly the
decrease in her sink capacity:
\[
 \min\{B_i-\lambda d_i,\lambda(R_i-d_i)\}
 =\min\{B_i,\lambda R_i\}-\lambda d_i.
\]
Thus the restricted flow is feasible and saturates all remaining
source arcs. Deleted buyers have zero remaining capacity and inflow.

To obtain a strict price increase, take an old residual path from a
remaining good to $t$. It cannot enter a closed good. It also cannot
pass through $s$, whose source arcs are all saturated. For a deleted
buyer, $\lambda d_i$ equals her old sink capacity; hence all her
inflow came from $S$ and her sink arc was saturated. Such a buyer
had no residual path to $t$, so the old path avoids her. At retained
buyers, the equal decreases in sink capacity and flow preserve
forward residual capacity, and every other arc on the path is
unchanged. All remaining goods can therefore still reach $t$.
Lemma~\ref{lem:threshold-price} gives a strictly larger feasible
price, completing the induction. Each stage removes at least one
good, so there are at most $N$ stages.
\end{proof}

\subsubsection{Completion and optimal demand}

\begin{lemma}[Completing the bundles]\label{lem:threshold-completion}
The completion step clears every good and fills every unit bundle
without changing utility or expenditure. The resulting allocation
respects all budgets, and prices are normalized.
\end{lemma}
\begin{proof}
Every allocation before completion respects the remaining budgets
and quantities. All goods outside $S_0$ are allocated in full.
Since both total supply and total required quantity are $N$, the
unused supply in $S_0$ equals $\sum_iR_i$. The greedy procedure
therefore exhausts all remaining supply and all deficits, preserving
nonnegativity and satisfying every row and column constraint.

A buyer with a remaining deficit also had a deficit after the free
stage. By Lemma~\ref{lem:threshold-free}, she has already received
her full threshold on each good of $S_0$. Additional quantities
from these goods have zero marginal utility and zero cost, including
for buyers whose budgets have been exhausted. Finally, $S_0$ is
nonempty and has price zero, so the prices are normalized.
\end{proof}

\begin{lemma}[Optimal demand and expenditure]
\label{lem:threshold-optimality}
Every buyer's final allocation is a minimum-expenditure bundle
among her utility-maximizing affordable unit bundles.
\end{lemma}
\begin{proof}
First consider a buyer who leaves the active set at a positive price
$P$. She has received every threshold below $P$, receives no goods
above $P$, and pays only for useful quantities. Completion adds
only free quantities and preserves these properties.

For an alternative unit bundle $y_i$, write
$z_j=\min\{y_{ij},\tau_{ij}\}$ and
$\bar z_j=\min\{x_{ij},\tau_{ij}\}$. Below $P$ we have
$\bar z_j=\tau_{ij}$ and $z_j-\bar z_j\le0$; above $P$ we
have $\bar z_j=0$ and $z_j-\bar z_j\ge0$. Including goods at
price $P$, we obtain in every case
\[
 p_j(z_j-\bar z_j)\ge P(z_j-\bar z_j).
\]
Since $p\cdot y_i\ge\sum_jp_jz_j$ and
$p\cdot x_i=\sum_jp_j\bar z_j$, summation gives
\begin{equation}\label{eq:threshold-demand}
 p\cdot y_i-p\cdot x_i
 \ge P\bigl(U_i(y_i)-U_i(x_i)\bigr).
\end{equation}
If her budget was exhausted, $p\cdot x_i=1$, and any strict
utility improvement would cost more than one. If her quantity was
exhausted, she already has utility one, the maximum possible on a
unit bundle. In either case she is optimal, and
\eqref{eq:threshold-demand} excludes a cheaper bundle of equal
utility.

A buyer who fills her useful bundle in the free stage has utility
one at zero cost. Finally, consider a buyer still active when all
prices have been assigned. The invariant gives her every threshold,
so her utility is $\sum_j\tau_{ij}$. This is the largest possible
utility, and any bundle attaining it must contain at least
$\tau_{ij}$ of every good, at cost at least
$\sum_jp_j\tau_{ij}$. Her allocation has exactly this cost,
because completion adds only free goods. This includes all-zero
threshold rows and thus the originally indifferent buyers.
\end{proof}

\subsection{Encoding lengths and strongly polynomial running time}

The stage prices depend on budgets obtained at earlier stages.
To bound their encoding lengths, we use the invariant to express
the surviving buyers' allocations in terms of the input thresholds.
We also bound all trial parameters and flow data used by the
parametric price finder.

\begin{lemma}[Polynomial encoding lengths]\label{lem:threshold-bits}
All trial parameters, capacities, flows, prices, remaining budgets
and quantities, and final allocations can be computed with encoding
length polynomial in the binary input length.
\end{lemma}
\begin{proof}
Let $Q$ be the product of the input threshold denominators, so
$\log Q$ is at most the binary input length. An active buyer has
received her exact threshold on every previously allocated good.
Her remaining quantity therefore has denominator dividing $Q$.

We first express an optimal stage price using a tight cut. For
$S\subseteq\mathcal J$ and $i\in\mathcal I$, put
$h_i(S)=\min\{R_i,\sum_{j\in S}\tau_{ij}\}$.
A cut whose source-side goods are exactly $S$ pays
$\lambda(|\mathcal J|-|S|)$ for source arcs. Placing buyer $i$
on the source side adds $\min\{B_i,\lambda R_i\}$; placing
her on the sink side adds $\lambda\sum_{j\in S}\tau_{ij}$.
Minimizing independently over buyers and applying max-flow/min-cut
shows that a positive price is feasible exactly when
\begin{equation}\label{eq:threshold-cut}
 \sum_{i\in\mathcal I}\min\{B_i/\lambda,h_i(S)\}\ge|S|
 \qquad\text{for every }S\subseteq\mathcal J.
\end{equation}
At the optimum $\lambda_*$, some inequality in~\eqref{eq:threshold-cut} is tight and
becomes false immediately to its right. Indeed, there are finitely
many sets $S$, their left sides are continuous and nonincreasing,
and every larger price violates at least one inequality. Fix such
a set. Let $A\subseteq\mathcal I$ consist of the buyers whose minimum is
$B_i/\lambda$ immediately to the right of $\lambda_*$, including
those whose terms tie at $\lambda_*$. The set $A$ is nonempty,
since otherwise the left side would remain constant. Tightness gives
\[
 \lambda_*=
 \frac{\sum_{i\in A}B_i}{|S|-\sum_{i\in\mathcal I\setminus A}h_i(S)}.
\]
Each $h_i(S)$ has denominator dividing $Q$. The denominator of
this ratio equals $\sum_{i\in A}B_i/\lambda_*>0$ and is at
most $|S|\le N$. Thus
\[
 d=Q\left(|S|-\sum_{i\in\mathcal I\setminus A}h_i(S)\right)
 \quad\text{is an integer with }1\le d\le NQ.
\]

Suppose the active budgets share denominator $E$, and write
$\sum_{i\in A}B_i=M/E$. Then $\lambda_*=QM/(Ed)$. By the
invariant, a surviving buyer receives $T_i/Q$ in this stage for
some integer $T_i$. If her old budget was $m_i/E$, her new budget is
\[
 \frac{m_i}{E}-\frac{QM}{Ed}\frac{T_i}{Q}
 =\frac{m_id-MT_i}{Ed}.
\]
The surviving budgets therefore share denominator $Ed$. Starting
from $E=1$, at most $N$ stages give
$\log E\le N\log(NQ)$. Budgets lie in $[0,1]$ and prices in
$[0,N]$, so their numerators also have polynomial bit length.
All data for each subsequent price search consequently have
polynomial encoding length.

We next bound the intermediate numbers used within a GGT call.
For a source-side cut $X$ in $\mathcal Q(\vartheta)$, write
\[
 c_\vartheta(X)=a_X\vartheta+b_X,\qquad
 a_X=\frac{M_X}{E},\qquad b_X=\frac{H_X}{Q}.
\]
Here $a_X$ is a sum of active budgets, and $b_X$ is a sum of
constant arc capacities. Thus $M_X,H_X$ are nonnegative integers
with $M_X\le NE$ and $H_X\le(N^2+2N)Q$.
Besides the endpoints $0,T$, the trial parameters of the GGT
breakpoint algorithm are intersections of affine cut capacities;
an intersection of two distinct slopes has the form
\[
 \vartheta=
 \frac{E(H_Y-H_X)}{Q(M_X-M_Y)}.
\]
The same description applies after contraction, since cuts of a
contracted network lift to cuts of the original stage network;
any common affine contribution cancels in the intersection.
The endpoint $T=\max_iR_i/B_i$ also has polynomial encoding
length. Hence every trial parameter and every trial capacity
has polynomial encoding length.
There are only polynomially many trials. The product of their
capacity denominators therefore has polynomial bit length and
is a common denominator for the flow and residual updates,
which use additions, subtractions, and comparisons.
The magnitudes of these quantities are bounded by sums of the
trial capacities. Thus all intermediate flow data in the GGT
call have polynomial encoding length as well.

At the accepted parameter $\vartheta_*=Ed/(QM)$, the source
capacities of $\mathcal Q(\vartheta_*)$ are
$B_i\vartheta_*=m_id/(QM)$, and all other capacities have
denominators dividing $Q$. The final ordinary maximum flow,
computed from zero by shortest augmenting paths, therefore has
entries $q_{ij}=z_{ij}/(QM)$ for integers $z_{ij}$.
Its expenditure entries are
$w_{ij}=\lambda_*q_{ij}=z_{ij}/(Ed)$.
The initial quantity flow likewise has denominator dividing $Q$.
Subtracting committed quantities and costs also gives
polynomial-length remaining quantities and budgets for departing
buyers.

For completion, take the product of the denominators of the
remaining supplies, deficits, and allocation entries. There are
polynomially many such quantities, each of polynomial encoding
length, so this common denominator has polynomial bit length.
The filling procedure uses only addition, subtraction, and minima;
it preserves this denominator and keeps all allocation entries in
$[0,1]$. The final allocation therefore has polynomial encoding
length as well.
\end{proof}

\begin{proof}[Proof of Theorem~\ref{thm:two-slopes}]
Lemmas~\ref{lem:threshold-free} and~\ref{lem:threshold-stage}
establish feasibility, the invariant, and termination.
Lemmas~\ref{lem:threshold-completion} and~\ref{lem:threshold-optimality}
give an exact normalized equilibrium with minimum-expenditure
tie-breaking. The initial normalization preserves these conclusions
for the original utilities.

The algorithm uses one initial maximum-flow computation and at most
$N$ positive-price stages. By Lemma~\ref{lem:threshold-strong-price},
each stage uses one GGT call, one ordinary maximum-flow computation,
and a residual reachability search in the expenditure network.
Every network has $v=O(N)$ vertices and $m=O(N^2)$ arcs.
The GGT breakpoint algorithm uses
$O(vm\log(2+v^2/m))$ arithmetic operations, while the
Edmonds--Karp shortest-augmenting-path algorithm uses $O(vm^2)$.
Thus the total is $O(N^6)$ arithmetic operations, without optimizing
the ordinary maximum-flow subroutine. Cut searches and allocation
updates have smaller cost. Completion makes at most $2N$ nonzero
assignments, because each exhausts a remaining supply or a deficit.
The operation count is independent of the magnitudes and encoding
lengths of the rational input, and Lemma~\ref{lem:threshold-bits}
bounds every intermediate encoding length polynomially.
The algorithm is therefore strongly polynomial.
\end{proof}

\subsection{Three common marginal slopes}

Theorem~\ref{thm:two-slopes} permits two common slopes for each
buyer. With three prescribed slopes $1,a,0$, the class includes
the hard linear instances of Theorem~\ref{thm:main}. As in
Corollary~\ref{cor:three-slopes}, segments may have zero length.
Approximation is defined by Definition~\ref{def:approx}, using
separable utility in the budget-one benchmark and imposing no
minimum-expenditure tie-breaking.

\begin{proof}[Proof of Corollary~\ref{cor:three-slopes}]
Fix a rational $a\in(0,1)$ and set
\begin{equation*}
 f_{ij}(t)=(1-a)\min\{t,\tau^{(1)}_{ij}\}
              +a\min\{t,\tau^{(2)}_{ij}\},
 \qquad 0\le\tau^{(1)}_{ij}\le\tau^{(2)}_{ij}\le1.
\end{equation*}
The successive marginal slopes are $1,a,0$. Given a linear
instance with utilities in $\{0,a,1\}$, choose
\[
 (\tau^{(1)}_{ij},\tau^{(2)}_{ij})=
 \begin{cases}
 (0,0),&u_{ij}=0,\\
 (0,1),&u_{ij}=a,\\
 (1,1),&u_{ij}=1.
 \end{cases}
\]
Then $f_{ij}(t)=u_{ij}t$ on $[0,1]$. Utilities of all unit
bundles and all budget-one demand values are unchanged, as are
supplies, budgets, and price normalization. A normalized additive
$\eps_*(a)/N$-equilibrium of the transformed instance is therefore
an equilibrium with the same error in the original instance.
The transformation has polynomial size and uses the identity
decoder, so Theorem~\ref{thm:main} proves the corollary.
Lemma~\ref{lem:rational-witness} gives existence and polynomial-length
rational witnesses, since the utilities on unit bundles are identical.
\end{proof}

\section{Discussion}

Our results distinguish two ways of extending the tractable
bi-valued HZ model. For linear utilities, adding a third common
value makes inverse-linear approximation PPAD-complete, for every
prescribed rational middle value $a\in(0,1)$. For separable
concave utilities, arbitrary rational buyer--good breakpoints
remain compatible with strongly polynomial-time exact computation when
each buyer has two marginal slopes common across goods.
These results identify a boundary in the expressive power of
preferences, while leaving several finer computational questions
open.

\paragraph{The optimal approximation error.}
For each fixed rational $a\in(0,1)$, what is the optimal
polynomial-time approximation guarantee for HZ markets with
utilities in $\{0,a,1\}$, under the normalized additive notion
of Definition~\ref{def:approx}?
Specializing the analysis of~\cite[Lemma~4.7]{yl} gives an
additive $a(1-a)$ guarantee. Since $a$ is fixed and rational,
the clone proportions $a$ and $1-a$ can be implemented exactly,
without rounding loss. In particular, the guarantee is $1/4$
for $\{0,1/2,1\}$.
Can this bound be improved, and what hardness threshold can
be established under PCP for PPAD?
Our constant-error lower bound leaves a gap to this algorithmic
guarantee. Determining the threshold for $a=1/2$ is a concrete
first target. More generally, how does the optimal error depend
on $a$, especially as $a$ approaches $0$ or $1$, where the
utility domain becomes bi-valued?

\paragraph{Few buyers with three utility values.}
Fix a rational $a\in(0,1)$, and suppose that only $k$ buyers
use all three values in $\{0,a,1\}$, while every other buyer
uses at most two.
Can a normalized additive $\eps$-equilibrium be computed in
time polynomial in $N$ and $1/\eps$ for every fixed $k$?
Here the polynomial degree may depend on $k$.
The case $k=1$ is a natural starting point.
A stronger goal is an algorithm with running time
\[
 f(k)\operatorname{poly}(N,1/\eps),
\]
where the polynomial degree is independent of $k$.
This asks whether the bi-valued algorithm can accommodate a
small number of buyers with richer preferences.
The remaining buyers may have arbitrarily many distinct
bi-valued preference types, so algorithms for a fixed total
number of preference types~\cite{ajt} do not resolve this
question.

\clearpage
\appendix

\section{Parameter bounds}\label{sec:parameter-checks}

This appendix verifies the integer multiplicities chosen in
Subsection~\ref{sec:choose-K}. Fix $a=A/D$ in lowest terms,
$C=16D^2$, $b=4C$, and $0<\tau\le\kappa/100$.
Let $K$ be a positive multiple of $D^2(D^2-A^2)$ satisfying
\eqref{eq:check-basic}. Every bound below holds for all six degree
pairs $h_v\in\{0,1,2\}$ and $o_v\in\{1,2\}$.

\begin{lemma}[Parameter feasibility]\label{lem:parameters}
All supplies and buyer counts are positive integers. The endpoints
satisfy $1<L<U$, and the populations satisfy
\[
 F_v>S_L+S_M,\qquad c_v>S_H,\qquad
 A^R_{Fv}>S^R_{Fv},\qquad A^R_{Uv}>S^R_{Uv}.
\]
The endpoints satisfy
\begin{equation*}
 aL\ge\frac{1+a}{2a},\qquad
 a(U+\eta)<\frac1a,\qquad L+\eta<U.
\end{equation*}
Consequently, $a(L+\eta)<1/a$ and $a^2(U+\eta)<1$.
The decoding error obeys
\begin{equation}\label{eq:final-checks}
 \frac{6a^2w(1+\tau)}{L-1}+2\tau<\frac\kappa{10},
 \qquad \frac\eta w=\tau<\frac\kappa4.
\end{equation}
\end{lemma}

\begin{lemma}[Supply and population inequalities]\label{lem:clearing-margins}
At each vertex, the low supply satisfies
\begin{equation}\label{eq:check-coarse}
 o_vK<S_L<o_vK+b_v,\qquad
 \frac{o_vK+b_v}{S_L}<aL.
\end{equation}
The link populations and middle supply satisfy
$c_v>b+1$, $aLS_M>c_v+10$, and
\begin{equation*}
 \frac{3(b+1)(1+a)}K<1.
\end{equation*}
The high supply satisfies $LS_H>10$ and $(L+\eta)S_H<c_v$.
In particular, $c_v>b_v$ and $aLS_M>b_v$.
\end{lemma}

\begin{proof}[Proof of Lemmas~\ref{lem:parameters} and~\ref{lem:clearing-margins}]
We first record the consequences of~\eqref{eq:check-basic} used
throughout the proof. Since $D\ge2$, we have $C\ge64$ and
$b\ge256$. The bound on $K$ gives
\[
 K>\max\left\{48,6(b+1),\frac{2a(b+1)}{1-a}\right\},
 \qquad \gamma K>4(b+1),\qquad w<\frac1{128}.
\]
Here
$\gamma K\ge(8b+40000)/((1-a^2)(1-a))>8b>4(b+1)$,
and we also have $K\ge200/(\kappa a^2)$.

\emph{Integrality and positivity.}
Write $K=jD^2(D^2-A^2)$. Then $\gamma K=jA^2D^2$,
so $C,b,c_v$ are multiples of $D^2$ and
\[
 S_L=o_vK+C,\qquad S_M=AD(jo_vA^2-48),\qquad S_H=32A^2
\]
are integers. The floor, reference, and anchor counts are
integers by their definitions using ceilings. Also,
\[
 b-5\le b_v\le b+1,\qquad
 c_v\ge\gamma K-b>3b+4>b+1.
\]
These inequalities imply positivity of every count and supply.
They also give $C<b-5\le b_v$, hence
$o_vK<S_L<o_vK+b_v$.

\emph{Price and clearing inequalities.}
For the remaining low-price comparison,
\[
 \frac{o_vK+b_v}{S_L}
 =1+\frac{b_v-C}{o_vK+C}
 <1+\frac{b+1}{K}
 <1+\frac{1-a}{2a},
\]
whereas
\[
 aL-1=\frac{1-a}{a}\left(1-\frac{3(1+a)}K\right)
 >\frac{1-a}{2a}.
\]
This proves~\eqref{eq:check-coarse} and
$aL>(1+a)/(2a)>1$. Since
$c_v+C=o_v\gamma K-3C<o_v\gamma K$,
\[
 \begin{aligned}
 aLS_M
 &=c_v+C-\frac{3(1-a^2)(c_v+C)}K\\
 &>c_v+C-3o_va^2
 >c_v+C-6\ge c_v+58>c_v+10.
 \end{aligned}
\]
The remaining link inequality follows from
$3(b+1)(1+a)/K<6(b+1)/K<1$.
Since $c_v>b+1\ge b_v$, we also have $c_v>b_v$ and $aLS_M>b_v$.

Using $S_H=2a^2C$,
\[
 LS_H=\left(1-\frac{3(1-a^2)}K\right)2C
 >\frac{15}{16}\,2C\ge120>10.
\]
Also $L+\eta=a^{-2}-(3/2-\tau)w<a^{-2}$, so
$(L+\eta)S_H<2C<c_v$.
Because $0<\tau<1/2$, we have $\eta<w$, $L+\eta<U$, and
\[
 a(U+\eta)=\frac1a-a(1/2-\tau)w<\frac1a.
\]
Together with $aL>(1+a)/(2a)$, these are the required reference
bounds and imply $1<L<U$. Since $aL,aU>1$, each anchor count
strictly exceeds its reference supply. Also $0<\theta<1$ gives
$F_v>S_L+S_M$, while $(L+\eta)S_H<c_v$ and $L+\eta>1$
give $c_v>S_H$.

\emph{Decoding error.}
The identity $L-1=(1-a^2)(K-3)/(a^2K)$ yields
\[
 \frac{6a^2w(1+\tau)}{L-1}+2\tau
 =\frac{12a^2(1+\tau)}{K-3}+2\tau
 <\frac{14a^2}{K}+2\tau.
\]
The last inequality uses $K\ge48$ and $\tau\le\kappa/100<1/100$.
Since $K\ge200/(\kappa a^2)$,
\[
 \frac{14a^2}{K}+2\tau
 \le\frac{14\kappa a^4}{200}+\frac{2\kappa}{100}
 <\frac{9\kappa}{100}<\frac\kappa{10}.
\]
Finally, $\eta/w=\tau\le\kappa/100<\kappa/4$.
The same lower bound on $K$ gives
$w\le\kappa(1-a^2)/100<\kappa/4$.
\end{proof}

\begin{lemma}[Balancing goods]\label{lem:balancing}
For a nonempty graph, $N>G_{\rm listed}$. Hence the balancing
group $\mathcal Z$ is nonempty, and the completed market has the
same number of buyers and goods.
\end{lemma}
\begin{proof}
Counting each wire population at its source gives
\begin{equation}\label{eq:N}
 N=\sum_{v\in V}\left(b_v+c_v+F_v+o_vK
                 +A^R_{Fv}+A^R_{Uv}\right).
\end{equation}
The number of goods before adding $\mathcal Z$ is
\begin{equation*}
 G_{\rm listed}=\sum_{v\in V}\left(S_L+S_M+S_H
                       +S^R_{Fv}+S^R_{Uv}\right).
\end{equation*}
At each vertex, Lemma~\ref{lem:parameters} gives
$F_v>S_L+S_M$, $c_v>S_H$, $A^R_{Fv}>S^R_{Fv}$, and
$A^R_{Uv}>S^R_{Uv}$. The remaining terms $b_v+o_vK$ in
$N$ are positive. Summing over a nonempty graph yields
$N>G_{\rm listed}$, so adding $N-G_{\rm listed}$ zero-utility
goods makes total supply equal to the number of buyers.
\end{proof}

\section{Quantitative estimates for robustness}\label{app:robustness}

We prove the demand estimates used in Lemma~\ref{lem:pcp-waste}
and the local decoding statement of Lemma~\ref{lem:pcp-local}.
Fix a normalized additive $\eps$-equilibrium of the market in
Section~\ref{sec:construction}, and equalize prices within each
group using Lemma~\ref{lem:copies}. Retain the notation
$w=U-L$, $\eta=\tau w$, and $(t_v,z_v)$ from~\eqref{eq:decoder}.
Retain the unapproved expenditure $\mathcal W_v$ from
Subsection~\ref{sec:approximation}, and let $B_j(v)$ contain
the vertices at undirected distance at most $j$ from $v$.

All implicit constants depend only on $a$, uniformly over the graph,
admissible $K$, and $0<\tau\le\kappa/100$.

\subsection{Demand estimates}

For buyer $i$, write $E_i^+$ and $E_i^0$ for expenditure on goods
with $u_{ij}>0$ and $u_{ij}=0$, respectively, and write $Q_i^0$
for her total quantity of goods with $u_{ij}=0$. For a demand dual
$(\beta_i,\mu_i)$, the reduced slack of good $j$ is
$\beta_i p_j+\mu_i-u_{ij}$.

\begin{lemma}[Allocated slack and demand deficits]
\label{lem:robust-demand-slack}
Choose an optimal demand dual $(\beta_i,\mu_i)$ for buyer $i$:
\[
 \beta_i\ge0,\qquad \beta_i p_j+\mu_i\ge u_{ij},\qquad
 \beta_i+\mu_i=V_i(p).
\]
Then $0\le\beta_i,\mu_i\le1$, and her allocated reduced slack
satisfies
\begin{equation}\label{eq:allocated-dual-slack}
 \sum_j(\beta_i p_j+\mu_i-u_{ij})x_{ij}
 =V_i(p)-U_i(x_i)+\beta_i(p\cdot x_i-1)
 \le2\eps.
\end{equation}
If $\beta_i>0$, then
\[
 E_i^0\le\frac{2\eps}{\beta_i},\qquad
 1-\frac{\eps}{\beta_i}\le E_i^+\le1+\eps.
\]
If $\mu_i>0$, then $Q_i^0\le2\eps/\mu_i$. For a population
of $M$ identical buyers using one common optimal dual, a group with
reduced slack $s>0$ receives at most $2M\eps/s$ quantity and
$2M\eps p/s$ expenditure, where $p$ is its price.
\end{lemma}
\begin{proof}
Normalization supplies a zero-price good. Its dual constraint gives
$\mu_i\ge0$, and $\beta_i+\mu_i=V_i(p)\le1$ bounds both
variables by one. Expanding the left side of
\eqref{eq:allocated-dual-slack} and using the unit row sum gives
the identity. The utility error contributes at most $\eps$, and
the budget error contributes at most $\beta_i\eps\le\eps$.

On unapproved goods the slack is $\beta_i p_j+\mu_i$, which
gives the bounds on $E_i^0$ and $Q_i^0$. On approved goods,
dual feasibility and total quantity at most one give
$U_i(x_i)\le\beta_iE_i^++\mu_i$. Comparing with
$U_i(x_i)\ge\beta_i+\mu_i-\eps$ proves
$E_i^+\ge1-\eps/\beta_i$. The same lower bound holds for total
expenditure, and both expenditures are at most $1+\eps$.
Finally, sum the allocated-slack bound over a population sharing
one dual and retain the contribution of the specified group.
\end{proof}

We will also use the following consequence when a group price is
not yet bounded: if its reduced slack satisfies
$s\ge c(p-p_0)>0$, with $c>0$ depending only on $a$ and
$p_0=O_a(1)$, then $p/s=O_a(1+1/(p-p_0))$.

\begin{lemma}[A uniform utility gap]\label{lem:robust-utility-gap}
There are positive rational constants $d_a,g_a,c_a$, depending
only on $a$, such that, for all sufficiently large admissible $K$
and $\eps\le c_a$, every group valued at one by some buyer has
price greater than $1+d_a$, and every buyer satisfies
$V_i(p)\le1-g_a$.
\end{lemma}
\begin{proof}
Put $r_*=(1+a)/(2a)>1$ and
$s=1-1/r_*=(1-a)/(1+a)>\theta$. The parameter bounds give
$aU>aL\ge r_*$.
Choose positive rational $d,c$ so that
\[
 \begin{aligned}
 c&<s,& \frac1{r_*}+c&<\frac1{1+d},\\
 \frac{d+c}{1-a}+3c&<s-\theta,&
 \frac{d+c}{1-a}&<\frac12.
 \end{aligned}
\]
These choices are possible because $r_*>1$ and $s>\theta$.
Assume $\eps\le c$.

A reference with price $r$, anchor count $A$, and supply $S$
satisfies $A/S\ge r_*$. If $r\le1$, every anchor has optimum
one and must receive at least $1-\eps$ reference quantity. This
contradicts $A(1-\eps)>S$. Hence $r>1$, the anchor optimum is
$1/r$, and supply feasibility gives
\[
 \frac1r\le\frac SA+\eps\le\frac1{r_*}+\eps,
\]
and the choice of $c$ implies $r>1+d$.

Suppose a middle price satisfies $m\le1+d$. Every floor buyer
then has optimum at least $1/(1+d)$. If her middle and reference
quantities are $X_M,X_R$, all her other goods have utility at most
$a$, and hence
\[
 X_M\ge\frac{1/(1+d)-a-\eps}{1-a}-X_R
 \ge s-\frac{d+\eps}{1-a}-3\eps>\theta.
\]
The second inequality uses
$X_R\le(1+\eps)/r_F\le1/r_*+3\eps$.
Summing $X_M>\theta$ over the floor population contradicts
$F_v\theta>S_L+S_M$. Therefore $m>1+d$.

If a high price were at most $1+d$, each upper buyer would require
at least $[1/(1+d)-a-\eps]/(1-a)>1/2$ high quantity. For large
$K$, the inequality $c_v>2S_H$ rules this out.

References, middle goods, and high goods exhaust the groups valued
at one. A budget-one bundle contains at most $1/(1+d)$ total
quantity from these groups; every other good has utility at most
$a$. Consequently
\[
 V_i(p)\le a+\frac{1-a}{1+d}=1-g,
 \qquad g=\frac{(1-a)d}{1+d}>0.
\]
Take $d_a=d$, $g_a=g$, and $c_a=c$.
\end{proof}

\subsection{Local price estimates}

The next three lemmas hold with common constants $K_a$ and $c_a>0$
under
\[
 K\ge K_a,\qquad K^2\eps\le c_a,\qquad 0\le\zeta\le c_a.
\]
Here $K$ is admissible and $0<\tau\le\kappa/100$, as above.
The proofs choose $K_a$ large enough and $c_a$ small enough,
depending only on $a$.

\begin{lemma}[Preliminary local prices]\label{lem:robust-preliminary}
Under the stated small-error conditions, if $\mathcal W_v\le\zeta$, then
\begin{align*}
 aL-O_a(\eps)&\le r_{Fv}
   \le a(L+\eta)+O_a(\eps)+\zeta/S^R_{Fv},\\
 aU-O_a(\eps)&\le r_{Uv}
   \le a(U+\eta)+O_a(\eps)+\zeta/S^R_{Uv}.
\end{align*}
These references lie in a fixed compact subinterval of $(1,\infty)$,
and, for some constant $c>0$,
\[
 \max\{r_{Fv},r_{Uv}\}\le1/a-c/K,
 \qquad r_{Uv}-r_{Fv}\ge c/K.
\]
The target prices satisfy
\[
 \begin{gathered}
 q_v\ge ar_{Fv}-O_a(\eps),\qquad
 m_v\ge r_{Fv}-O_a(\eps),\\
 p_v\ge r_{Fv}/a-O_a(K\eps)\ge L-O_a(K\eps),\\
 a^2L-O_a(K\eps)\le q_v\le1+O_a(1/K)<aL-c.
 \end{gathered}
\]
An optimal dual shared by all floor buyers satisfies
\[
 \mu=O_a(\eps),\quad
 s_R=\beta r_{Fv}+\mu-1=O_a(\eps),\quad
 \beta=1/r_{Fv}+O_a(\eps)\ge c.
\]
No bounds on neighboring endpoint prices are required.
\end{lemma}
\begin{proof}
Suppress the vertex subscript.

For a reference with price $r$, supply $S$, anchor count $A$,
and $M$ nonanchor approvers, anchor utility and revenue clearing give
\[
 r\ge\frac{A/S}{1+\eps A/S},\qquad
 rS\le(A+M)(1+\eps)+\zeta.
\]
Substitute the counts in~\eqref{eq:reference-counts} to obtain
the two reference intervals. Since $1/S^R=O_a(\eta/K)$,
$\eta=\tau w$, and $w=\Theta_a(1/K)$, their errors are smaller
than the endpoint gaps in~\eqref{eq:interval} when $K^2\eps$
and $\zeta$ are sufficiently small. This proves the remaining
reference bounds and also gives
\[
 1/r_F\ge a+c/K,\qquad 1-1/r_F\ge\theta+c
\]
for a possibly smaller $c>0$.

Next average the $F$ identical floor rows and denote their quantities by
$\overline X_L,\overline X_M,\overline X_R,\overline X_0$.
Supply feasibility gives
\[
 \overline X_L\le S_L/F,\quad
 \overline X_M\le S_M/F,\quad
 \overline X_L+\overline X_M\le\theta.
\]
Both target supplies are of order $K$, with
$S_L/(S_L+S_M)$ bounded below by a positive constant depending
only on $a$. The separate supply bounds and $\theta\le a$ imply
\[
 a\overline X_L+\overline X_M
 \le\theta\frac{aS_L+S_M}{S_L+S_M}\le a-c.
\]
The optimum of a floor buyer is at least $1/r_F\ge a+c/K$,
so approximate optimality requires $\overline X_R\ge c$ after
reducing $c$ if necessary. On the other hand, the budget gives
$\overline X_R\le(1+\eps)/r_F=a+O_a(1/K+\eps)$.
Together with $\theta\le(1-a)/2$ and the unit row sums, this
implies $\overline X_0\ge(1-a)/3$ for large $K$.

Use one optimal dual for the whole floor population. By
Lemma~\ref{lem:robust-demand-slack}, the positive average reference
quantity forces $s_R=O_a(\eps)$, and the positive average
unapproved quantity forces $\mu=O_a(\eps)$. Consequently
$\beta=1/r_F+O_a(\eps)\ge c$. The low and middle dual
constraints now give $q\ge ar_F-O_a(\eps)$ and
$m\ge r_F-O_a(\eps)$.

To bound the high price, note that all groups approved by an upper
buyer cost more than one by Lemma~\ref{lem:robust-utility-gap}.
We may therefore choose an optimal dual with
\[
 \mu_U=0,\qquad \beta_U=\max\{1/p,a/m,a/r_U\}.
\]
The reference bound gives $\beta_U\ge c$, hence approved
expenditure $1+O_a(\eps)$ per buyer. Suppose
$\Delta=r_F/a-p>H_aK\eps$ for a sufficiently large constant
$H_a$. Then $m\ge r_F-O_a(\eps)$ and $r_U>r_F$ imply
$ap<m,r_U$. High goods uniquely maximize utility per dollar,
and each alternative's reduced slack is bounded below by a constant
times its price minus $ap$. The population estimate in
Lemma~\ref{lem:robust-demand-slack} bounds total expenditure on
the alternatives by
\[
 O_a\bigl(K\eps(1+1/\Delta)\bigr)
 =O_a(K\eps+1/H_a)=O_a(1).
\]
The price-to-slack estimate above applies even if $m$ is unbounded.
The upper population would therefore spend $c_v-O_a(1)$ on high
goods, whereas their revenue is $pS_H=O_a(1)$. This is impossible
for large $K$, proving $p\ge r_F/a-O_a(K\eps)$.

It remains to bound $q$ from above. The floor dual satisfies
\[
 s_L=\beta(q-ar_F)+(1-a)\mu+a s_R\ge\beta(q-ar_F).
\]
If $q>ar_F+H_aK\eps$, Lemma~\ref{lem:robust-demand-slack}
bounds floor expenditure on $L_v$ by $O_a(K\eps+1/H_a)$.
Choose $H_a$ large and $K^2\eps$ small so that this is at most
one. Only lower links and outgoing wires supply other approved
expenditure, giving
\[
 qS_L\le(o_vK+b_v)(1+\eps)+1+\zeta.
\]
If instead $q\le ar_F+H_aK\eps$, the reference upper bound
and smallness of $K^2\eps$ give $q<1$. In both cases,
$q\le1+O_a(1/K)$, and the lower bound on $q$ gives
$q\ge a^2L-O_a(K\eps)$. Since $aL\to1/a>1$, we also have
$q<aL-c$ for large $K$.
\end{proof}

\begin{lemma}[Approximate price relations]\label{lem:robust-price-chain}
Under the stated small-error conditions, suppose $\mathcal W_u\le\zeta$
for $u=v$ and every incoming neighbor $u$ of $v$. Then
\[
 \begin{gathered}
 |q_v-a^2p_v|+|m_v-ap_v|=O_a(K\eps),\\
 L-O_a(K\eps)\le p_v\le U+2\eta+O_a(K\eps),\qquad q_v<1.
 \end{gathered}
\]
All three target prices consequently lie in fixed compact intervals.
\end{lemma}
\begin{proof}
Apply Lemma~\ref{lem:robust-preliminary} at $v$ and its incoming
neighbors. We first bound the expenditure of incoming wires on
$H_v$; this forces the upper population to purchase high goods.

For an incoming wire, write $q$ for the source low price and $p$
for the destination high price. The preliminary bounds give
\[
 q\le1+O_a(1/K),\quad p\ge a^{-2}-O_a(1/K),\quad
 q/p\le a^2+O_a(1/K)<a-c
\]
for some $c>0$. Thus the low option has strictly greater utility
per dollar, even when $p$ is unbounded. If $q\ge1$, its optimal
dual is $(a/q,0)$, and the high slack is at least $cp$; high
spending per buyer is $O_a(\eps)$. If $q<1$, the mixed dual is
\[
 \beta=\frac{1-a}{p-q},\qquad
 \mu=\frac{ap-q}{p-q}\ge c.
\]
The buyer's unapproved quantity is $O_a(\eps)$. Writing $Y$
for her high quantity, the unit row and budget bounds give
\[
 pY\le\frac{p}{p-q}\bigl(1-q+O_a(\eps)\bigr)
 \le\frac{L-O_a(K\eps)}{L-1-O_a(K\eps)}
       \bigl(1-a^2L+O_a(K\eps)\bigr).
\]
There are at most $2K$ incoming buyers. Multiplying this bound
by $2K$ gives a leading term tending to six. Including the
$q\ge1$ case, their total high expenditure is at most
$9+O_a(K^2\eps)$.

High revenue is at least $LS_H-O_a(K\eps)>10$. Since at most
$\zeta$ of it is unapproved expenditure, upper buyers must spend
at least $1/2$ on high goods. Choose their common optimal dual
$\mu_U=0$, $\beta_U=\max\{1/p,a/m,a/r_U\}$. Their high
expenditure multiplied by $\beta_U-1/p$ is its contribution to
allocated slack. Therefore
\[
 \beta_U-1/p=O_a(c_v\eps)=O_a(K\eps).
\]
Since $\beta_U\ge a/r_U\ge c$, this bounds $p$ from above.
Substituting this bound back into the inequality gives
$p\le r_U/a+O_a(K\eps)$ and $m\ge ap-O_a(K\eps)$.

Suppose $\Delta=m-ap>H_aK\eps$, for a sufficiently large $H_a$.
The upper buyers' middle expenditure is
$O_a(K\eps+1/H_a)$, and is at most one for large $H_a$
and small $K^2\eps$. The lower bound on $p$ gives
$m-r_F\ge\Delta-O_a(K\eps)$, so the floor population's
middle expenditure is also at most one. Lower links contribute at
most $b_v(1+\eps)$ and unapproved buyers at most $\zeta$.
These cannot cover
\[
 mS_M\ge aLS_M-O_a(K^2\eps)>c_v+10-O_a(K^2\eps),
\]
using $c_v>b+1\ge b_v$. Thus $|m-ap|=O_a(K\eps)$.

We next compare $q$ with $am$. Suppose first that
$\Delta=q-am>H_aK\eps$. An optimal lower-link bundle uses
middle goods and filler, and its dual gives low slack
$\Delta/m\ge c\Delta$.
Since $m\ge r_F-O_a(K\eps)$, the floor low slack has the same
lower bound after increasing $H_a$. Their combined low quantity is
at most one, whereas unapproved buyers receive at most
$\zeta/q=O_a(\zeta)$ low quantity. Outgoing wires receive at
most $o_vK$. This contradicts $S_L=o_vK+C$ and $C\ge64$. Hence
$q\le am+O_a(K\eps)$. Together with
$m=ap+O_a(K\eps)$ and $p\le r_U/a+O_a(K\eps)$, this
gives $q<1$ for sufficiently small $K^2\eps$.

For the reverse bound, suppose $\Delta=am-q>H_aK\eps$. The
lower bound $q\ge ar_F-O_a(\eps)$ implies
$m-r_F\ge\Delta/a-O_a(\eps)$, so floor middle expenditure
is at most one. For a lower link, $q<1<m$ and its mixed dual is
\[
 \beta=\frac{1-a}{m-q}\ge c,\qquad
 \mu=\frac{am-q}{m-q}\ge c\Delta.
\]
Its unapproved quantity is $O_a(\eps/\Delta)$. The row and
budget bounds therefore limit its middle expenditure to
\[
 \frac{m}{m-q}\bigl(1-q+O_a(\eps+\eps/\Delta)\bigr)
 \le\frac{1}{1-a}\bigl(1-q+O_a(\eps+\eps/\Delta)\bigr)
 =O_a(1/K+K\eps+\eps/\Delta).
\]
There are only $b_v=O_a(1)$ lower links, so this is below one in
total for large $K$. Upper buyers spend at most $c_v(1+\eps)$
and unapproved buyers at most $\zeta$. Together, these expenditures
are less than the required revenue $mS_M>c_v+10-O_a(K^2\eps)$,
a contradiction. Thus $|q-am|=O_a(K\eps)$.

Combining this bound with $|m-ap|=O_a(K\eps)$ gives
$|q-a^2p|=O_a(K\eps)$. The reference bound and
$1/S^R_U=O_a(\eta/K)$ imply
$r_U/a\le U+\eta+O_a(\eps)+O_a(\zeta\eta/K)
\le U+2\eta+O_a(\eps)$.
Together with $r_F/a-O_a(K\eps)\le p\le r_U/a+O_a(K\eps)$,
this gives the stated interval for $p_v$.
\end{proof}

\begin{lemma}[Approximate local accounting]\label{lem:robust-accounting}
Suppose $\mathcal W_u\le\zeta$ for every $u\in B_2(v)$. Under the
stated small-error conditions, evaluate the rational wire
payments \eqref{eq:wire-dollars} and residual \eqref{eq:F} at
the ideal low and middle prices $a^2p_u,ap_u$, keeping the actual
high prices $p_u$. Denote this residual by $\Phi_v(p)$ and
abbreviate it to $\Phi_v$ in the estimates below.
Let $W_v$ and $V_v$ be the actual approved expenditures of floor
and upper buyers, respectively, on the three targets of $v$.
Then
\begin{align}
 \Phi_v&=W_v+(V_v-c_v)+O_a(K^2\eps+\zeta),
       \label{eq:robust-accounting}\\
 \left|\frac{\Phi_v}{2}
       -\left(\sum_{u\to v}z_u-\frac12\right)\right|
 &\le O_a(1/K+\tau+K^2\eps).
       \label{eq:robust-threshold-expansion}
\end{align}
\end{lemma}
\begin{proof}
Every incoming neighbor of $v$ or of a neighbor of $v$ lies in
$B_2(v)$. Lemma~\ref{lem:robust-price-chain} therefore applies
at both endpoints of every wire incident to $v$.

For each such wire, both optimal dual multipliers are bounded below
by positive constants. Lemma~\ref{lem:robust-demand-slack} gives
approved quantity and expenditure $1+O_a(\eps)$ per buyer.
Solving for the low and high quantities then gives expenditures
within $O_a(\eps)$ of the exact mixed-demand payments at the
actual endpoint prices. The denominators are bounded away from
zero. Replacing the source low price $q_u$ by $a^2p_u$ changes
each buyer's payment by $O_a(K\eps)$, and the whole population's
payment by $O_a(K^2\eps)$.

The actual target revenue is
$q_vS_L+m_vS_M+p_vS_H$. Replacing $q_v,m_v$ by $a^2p_v,ap_v$
changes this by $O_a(K^2\eps)$ and yields the ideal revenue
$R_v$ used in \eqref{eq:F}. Each lower link has a bounded-price
utility-one option and the uniform utility gap, so its dual
$\beta$ is bounded below. Its approved expenditure is therefore
$1+O_a(\eps)$; the whole lower population contributes
$b_v+O_a(\eps)$. Actual column clearing, the wire comparison,
and at most $\zeta$ unapproved target expenditure now give
\eqref{eq:robust-accounting}.

To prove the threshold expansion, apply identity~\eqref{eq:canceled}
to these prices. The approximate price relations give
\[
 -O_a(K^2\eps)\le t_u\le1+2\tau+O_a(K^2\eps),\qquad
 |t_u-z_u|=O_a(\tau+K^2\eps).
\]
For each incident edge, its ideal payment is
$I_{uv}=3-2t_u+e_{uv}$, with
\[
 e_{uv}=(3-2t_u)\frac{a^2(p_u-p_v)}{p_v-a^2p_u}.
\]
The extended price interval gives $p_u-p_v=O_a(w)$ and
$3-2t_u=O_a(1)$, while $p_v-a^2p_u$ is bounded away from
zero. Thus $e_{uv}=O_a(1/K)$. There are at most four incident
edge contributions to $e_v$, and~\eqref{eq:canceled} becomes
$\Phi_v=2\sum_{u\to v}t_u-1-e_v$. The displayed clipping
bound now proves~\eqref{eq:robust-threshold-expansion}.
\end{proof}

\subsection{Proof of the local decoding lemma}
\label{app:local-decoding-proof}

\begin{proof}[Proof of Lemma~\ref{lem:pcp-local}]
The hypothesis $\mathcal W_u\le\zeta$ on $B_2(v)$ supplies all local
estimates used below. Lemma~\ref{lem:robust-preliminary} only
requires the module at its vertex; Lemma~\ref{lem:robust-price-chain}
also requires the preliminary bounds at incoming neighbors.
Consequently, the price relations hold at $v$ and at both endpoints
of each incident wire, and Lemma~\ref{lem:robust-accounting}
applies. No price bounds outside $B_2(v)$ are used.

Fix a constant $C_a$ large enough to dominate the error constants
in these lemmas and in the slack estimates below. Take $c_a$ small
enough to satisfy their hypotheses, the endpoint-gap estimates
below, and $c_a\le1/(24C_a)$. Increase $K_a$ so that
$K_a\ge8C_a/\kappa$. Under the hypotheses of the local decoding
lemma,
\[
 C_a\left(\frac1K+\tau+\frac{K^2\eps}{\kappa}+\zeta\right)
 \le\frac\kappa4<\frac\kappa2.
\]
In particular, the error in
\eqref{eq:robust-threshold-expansion} is less than $\kappa/2$.
We prove the two threshold implications separately.

Suppose $\sum_{u\to v}z_u>1/2+\kappa$ and $z_v>\kappa$.
The threshold expansion gives $\Phi_v>\kappa$, while the decoder
gives $p_v>L+\kappa w$. The price relations and reference upper
bound imply
\[
 \begin{aligned}
 q_v-ar_{Fv}
 &\ge a^2(\kappa w-\eta)-O_a(K\eps+\zeta\eta/K),\\
 m_v-r_{Fv}
 &\ge a(\kappa w-\eta)-O_a(K\eps+\zeta\eta/K).
 \end{aligned}
\]
Because $\eta=\tau w$ and $w=\Theta_a(1/K)$, the assumptions
$\tau\le c_a\kappa$, $K^2\eps\le c_a\kappa^2$, and
$\zeta\le c_a\kappa$ make both differences at least a positive
constant times $\kappa w$ when $c_a$ is sufficiently small.
For the floor dual supplied by Lemma~\ref{lem:robust-preliminary},
$\beta\ge c>0$, $\mu\ge0$, and $s_R\ge0$. The exact identities
\[
 s_L=\beta(q_v-ar_{Fv})+(1-a)\mu+a s_R,\qquad
 s_M=\beta(m_v-r_{Fv})+s_R
\]
therefore give $s_L,s_M\ge c\kappa w$.
Lemma~\ref{lem:robust-demand-slack}, the bounded target prices,
and $F_v=O_a(K)$ imply
\[
 W_v=O_a(K^2\eps/\kappa),\qquad
 V_v-c_v\le c_v\eps=O_a(K\eps).
\]
Substitution in~\eqref{eq:robust-accounting} gives
$\Phi_v\le C_a(K^2\eps/\kappa+\zeta)<\kappa/2$,
contradicting $\Phi_v>\kappa$. This proves the first implication.

Now suppose $\sum_{u\to v}z_u<1/2-\kappa$ and
$z_v<1-\kappa$. Then $\Phi_v<-\kappa$ and
$p_v<U-\kappa w$. The lower reference bound gives
$r_{Uv}-ap_v\ge a\kappa w-O_a(\eps)$.
Choose the upper population's optimal dual as
$\mu_U=0$ and $\beta_U=\max\{1/p_v,a/m_v,a/r_{Uv}\}$.
Its reference slack satisfies
\[
 \beta_Ur_{Uv}-a\ge
 \frac{r_{Uv}-ap_v}{p_v}\ge c\kappa w.
\]
The upper population therefore spends $O_a(K^2\eps/\kappa)$
on its reference. Its total approved expenditure is
$c_v+O_a(K\eps)$ by Lemma~\ref{lem:robust-demand-slack}, so
\[
 V_v-c_v\ge-O_a(K^2\eps/\kappa),\qquad W_v\ge0.
\]
Equation~\eqref{eq:robust-accounting} now gives
$\Phi_v\ge-C_a(K^2\eps/\kappa+\zeta)>-\kappa/2$,
contradicting $\Phi_v<-\kappa$. This proves the second implication.
\end{proof}

\clearpage

\section{Rational approximate witnesses}\label{app:rational-witness-proof}

\begin{proof}[Proof of Lemma~\ref{lem:rational-witness}]
Take a normalized exact equilibrium $(p,x)$ with minimum-expenditure
tie-breaking, supplied by \cite[Theorem 5.1, conference version]{gtv}.
Clearing and the budget constraints give $\sum_jp_j\le N$.
Put $Q=\lceil(N+1)/\eps\rceil$ and
$p'_j=\lfloor Qp_j\rfloor/Q$. A zero-price good remains free.

Round $x$ to a doubly stochastic matrix $x'$ with entries in
$Q^{-1}\mathbb Z$ and $|x'_{ij}-x_{ij}|\le1/Q$. Such a matrix
exists by integrality of the transportation problem with row and
column sums $Q$ and bounds
$\lfloor Qx_{ij}\rfloor\le y_{ij}\le\lceil Qx_{ij}\rceil$:
$Qx$ is feasible, so an integral feasible $y$ gives $x'=y/Q$.
The rounding errors satisfy
\[
 p'\cdot x'_i\le1+\frac NQ,
 \qquad U_i(x'_i)\ge V_i(p)-\frac NQ.
\]
A unit bundle affordable at $p'$ costs at most $1+1/Q$ at $p$.
Scaling it by $1/(1+1/Q)$ and adding free filler loses at most
$1/Q$ utility. Hence $V_i(p')\le V_i(p)+1/Q$, and
\[
 U_i(x'_i)\ge V_i(p')-\frac{N+1}{Q}.
\]
Both errors are at most $\eps$. Row and column sums remain exact,
and prices remain normalized. Finally, $Q=O(N/\eps)$ and
$p'_j\le N$, which gives the bit bound.
\end{proof}

\section*{AI methodology statement}
\addcontentsline{toc}{section}{AI methodology statement}

GPT-6 Astra helped simplify and calibrate the construction in
Sections~\ref{sec:construction}--\ref{sec:exact},
refine the rational-witness argument in
Appendix~\ref{app:rational-witness-proof},
and improve the algorithm presentation and running-time analysis in
Section~\ref{sec:two-slopes}.
It also helped review the proofs, including the conditional hardness
result in Section~\ref{sec:pcp-equivalence}, check references,
and improve exposition, notation, and figures.
The authors take full responsibility for the correctness,
originality, and integrity of the paper, including its references.


\begin{thebibliography}{OpenAI26}
\normalsize
\setlength{\itemsep}{2pt}
\bibitem[AJT17]{ajt}
Saeed Alaei, Pooya Jalaly Khalilabadi, and \'Eva Tardos.
\newblock Computing equilibrium in matching markets.
\newblock In \emph{Proceedings of the 2017 ACM Conference on Economics
and Computation (EC)}, pages 245--261, 2017.
\newblock \url{https://doi.org/10.1145/3033274.3085150}.
\newblock Full version: \url{https://arxiv.org/abs/1703.10689}.

\bibitem[BPR16]{bpr}
Yakov Babichenko, Christos H. Papadimitriou, and Aviad Rubinstein.
\newblock Can almost everybody be almost happy?
\newblock In \emph{Proceedings of the 2016 ACM Conference on Innovations
in Theoretical Computer Science (ITCS)}, pages 1--9, 2016.
\newblock \url{https://doi.org/10.1145/2840728.2840731}.
\newblock Full version: \url{https://arxiv.org/abs/1504.02411}.

\bibitem[BLXZ26]{blxz}
Mark Braverman, Jingyi Liu, Eric Xue, and Chenghan Zhou.
\newblock Hardness of approximate Hylland--Zeckhauser equilibria.
\newblock Preprint, arXiv:2606.00951v2, June 9, 2026.
\newblock \url{https://arxiv.org/abs/2606.00951v2}.

\bibitem[CCPY22]{ccpy}
Thomas Chen, Xi Chen, Binghui Peng, and Mihalis Yannakakis.
\newblock Computational hardness of the Hylland--Zeckhauser scheme.
\newblock In \emph{Proceedings of the 2022 Annual ACM--SIAM Symposium
on Discrete Algorithms (SODA)}, pages 2253--2268, 2022.
\newblock \url{https://doi.org/10.1137/1.9781611977073.90}.
\newblock Full version: \url{https://arxiv.org/abs/2107.05746}.

\bibitem[DFHM24]{purecircuit}
Argyrios Deligkas, John Fearnley, Alexandros Hollender, and
Themistoklis Melissourgos.
\newblock Pure-Circuit: Tight inapproximability for PPAD.
\newblock \emph{Journal of the ACM}, 71(5), Article 31, 2024.
\newblock \url{https://doi.org/10.1145/3678166}.
\newblock Theorem~4.23 in the author version:
\href{https://repository.essex.ac.uk/38497/1/Pure-Circuit_Tight_Inapproximability_for_PPAD.pdf}{author-version PDF}.

\bibitem[DFHM26]{pcpmarkets}
Argyrios Deligkas, John Fearnley, Alexandros Hollender, and
Themistoklis Melissourgos.
\newblock Fisher markets with approximately optimal bundles and the
need for a PCP theorem for PPAD.
\newblock In \emph{Proceedings of the 58th Annual ACM Symposium on
Theory of Computing (STOC)}, pages 1049--1060, 2026.
\newblock \url{https://doi.org/10.1145/3798129.3800817}.
\newblock Full version: \url{https://arxiv.org/abs/2604.27276v3}.

\bibitem[EMZ21]{emz}
Federico Echenique, Antonio Miralles, and Jun Zhang.
\newblock Constrained pseudo-market equilibrium.
\newblock \emph{American Economic Review}, 111(11):3699--3732, 2021.
\newblock \url{https://doi.org/10.1257/aer.20201769}.

\bibitem[GGT89]{ggt}
Giorgio Gallo, Michael D. Grigoriadis, and Robert E. Tarjan.
\newblock A fast parametric maximum flow algorithm and applications.
\newblock \emph{SIAM Journal on Computing}, 18(1):30--55, 1989.
\newblock \url{https://doi.org/10.1137/0218003}.

\bibitem[GTV22]{gtv}
Jugal Garg, Yixin Tao, and L\'aszl\'o A. V\'egh.
\newblock Approximating equilibrium under constrained piecewise linear
concave utilities with applications to matching markets.
\newblock In \emph{Proceedings of the 2022 Annual ACM--SIAM Symposium
on Discrete Algorithms (SODA)}, pages 2269--2284, 2022.
\newblock \url{https://doi.org/10.1137/1.9781611977073.91}.
\newblock Conference version:
\url{https://www.laszlovegh.eu/papers/Garg-Tao-Vegh-SODA.pdf}.

\bibitem[HMPY18]{hmpy}
Yinghua He, Antonio Miralles, Marek Pycia, and Jianye Yan.
\newblock A pseudo-market approach to allocation with priorities.
\newblock \emph{American Economic Journal: Microeconomics},
10(3):272--314, 2018.
\newblock \url{https://doi.org/10.1257/mic.20150259}.

\bibitem[HV22]{HV}
Mojtaba Hosseini and Vijay V. Vazirani.
\newblock Nash-bargaining-based models for matching markets: One-sided
and two-sided; Fisher and Arrow-Debreu.
\newblock In \emph{13th Innovations in Theoretical Computer Science
Conference (ITCS 2022)}, volume 215 of LIPIcs, pages 86:1--86:20, 2022.
\newblock \url{https://doi.org/10.4230/LIPIcs.ITCS.2022.86}.

\bibitem[HZ79]{hz}
Aanund Hylland and Richard Zeckhauser.
\newblock The efficient allocation of individuals to positions.
\newblock \emph{Journal of Political Economy}, 87(2):293--314, 1979.
\newblock \url{https://doi.org/10.1086/260757}.

\bibitem[OpenAI26]{openaipcp}
OpenAI.
\newblock The PCP-for-PPAD conjecture: a quasilinear reduction.
\newblock Preprint, September 25, 2026.
\newblock \href{https://github.com/openai/math/blob/main/preprints/The-PCP-for-PPAD-conjecture-a-quasilinear-reduction-September-25-2026/paper.pdf}{Preprint PDF}.

\bibitem[PTV26]{ptv}
Ioannis Panageas, Thorben Tr\"obst, and Vijay V. Vazirani.
\newblock Time-efficient algorithms for Nash-bargaining-based matching
market models.
\newblock In \emph{Web and Internet Economics (WINE 2024)},
volume 15534 of Lecture Notes in Computer Science, pages 119--136,
Springer, 2026.
\newblock \url{https://doi.org/10.1007/978-3-032-08560-3_7}.
\newblock Full version: \url{https://arxiv.org/abs/2106.02024v6}.

\bibitem[PP21]{thresholdgames}
Christos H. Papadimitriou and Binghui Peng.
\newblock Public goods games in directed networks.
\newblock In \emph{Proceedings of the 22nd ACM Conference on Economics
and Computation (EC)}, pages 745--762, 2021.
\newblock \url{https://doi.org/10.1145/3465456.3467616}.
\newblock Full version: \url{https://arxiv.org/abs/2106.00718}.

\bibitem[TV25]{TV}
Thorben Tr\"obst and Vijay V. Vazirani.
\newblock Cardinal-utility matching markets: The quest for
envy-freeness, Pareto-optimality, and efficient computability.
\newblock \emph{Mathematics of Operations Research}, Articles in
Advance, pages 1--23, published online September 26, 2025.
\newblock \url{https://doi.org/10.1287/moor.2024.0770}.
\newblock Author manuscript:
\url{https://ics.uci.edu/~vazirani/MOR-EFPO.pdf}.

\bibitem[VY25]{vy}
Vijay V. Vazirani and Mihalis Yannakakis.
\newblock Computational complexity of the Hylland--Zeckhauser mechanism
for one-sided matching markets.
\newblock \emph{SIAM Journal on Computing}, 54(2):193--232, 2025.
\newblock \url{https://doi.org/10.1137/23M157586X}.
\newblock Author manuscript:
\url{https://ics.uci.edu/~vazirani/VY.pdf}.

\bibitem[YL26]{yl}
Yonglei Yan and Zhengyang Liu.
\newblock Approximation and irrationality in Hylland--Zeckhauser equilibria.
\newblock Preprint, arXiv:2606.06317v2, June 17, 2026.
\newblock \url{https://arxiv.org/abs/2606.06317v2}.

\end{thebibliography}
\end{document}